\documentclass[11pt]{article}

\usepackage[margin=0.875in]{geometry}
\usepackage{amsmath,amssymb,amsthm,mathtools}
\usepackage{booktabs}
\usepackage{enumitem}
\usepackage{float}
\usepackage{microtype}
\usepackage{tikz}
\usepackage[colorlinks=true,linkcolor=blue,citecolor=blue,urlcolor=blue]{hyperref}

\newtheorem{theorem}{Theorem}
\newtheorem{corollary}[theorem]{Corollary}
\newtheorem{proposition}[theorem]{Proposition}
\newtheorem{lemma}[theorem]{Lemma}
\newtheorem{remark}[theorem]{Remark}
\newtheorem{definition}[theorem]{Definition}

\newcommand{\R}{\mathbb R}
\newcommand{\Z}{\mathbb Z}
\newcommand{\T}{\mathbb T}

\newcommand{\dd}{\,\mathrm d}
\newcommand{\tr}{\operatorname{tr}}

\newcommand{\eps}{\varepsilon}
\newcommand{\av}[1]{\left\langle #1\right\rangle}
\newcommand{\Ical}{\mathcal I}
\newcommand{\Qcal}{\mathcal Q}
\newcommand{\Dcal}{\mathcal D}

\newcommand{\Hmat}{\mathsf H}

\title{A Hexagonal Counterexample to Log-Convexity of Fisher Information Along the Heat Flow}
\author{Jiayang Zou\textsuperscript{1,2}, Luyao Fan\textsuperscript{2}, Jiayang Gao\textsuperscript{2}, and Jia Wang\textsuperscript{2}\\[0.5ex]
\small \textsuperscript{1}Stanford University, Stanford, CA, USA; \texttt{jyangzou@stanford.edu}\\
\small \textsuperscript{2}Shanghai Jiao Tong University, Shanghai, China; \texttt{\char123 qiudao, fanluyao, gjy0515, jiawang\char125 @sjtu.edu.cn}}
\date{}

\begin{document}
\maketitle

\begin{abstract}
We construct a smooth, strictly positive, Gaussian-decaying density on $\R^2$ for which Fisher information along the heat flow is not log-convex.  This disproves the Cheng--Geng log-convexity conjecture in dimension two and, by tensorization, in every dimension $d\ge2$.  Consequently, the multidimensional forms of the Gaussian completely monotone conjecture, McKean's conjecture, and Toscani's entropy power conjecture also fail, complementing the one-dimensional counterexample of Gu and Sellke.  Our construction is a small hexagonal perturbation on the triangular torus, transferred to $\R^2$ by a Gaussian envelope and supported by explicit two-dimensional numerics.  We also initiate the study of the sharp constants $\theta_d^*$ by proving $\theta_1^*=1$, establishing monotonicity in the dimension, and identifying a dichotomy for the asymptotic constant $\theta_\infty^*$ governed by the sign of $\Dcal$.  The explicit two-dimensional counterexample was obtained through GPT-5.5 Pro-assisted exploration.
\end{abstract}

\section{Introduction}

Let $(P_t)_{t\ge0}$ denote the heat semigroup on $\R^d$ with the probability convention
\begin{equation}\label{eq:heat}
        P_t=e^{t\Delta/2},
        \qquad
        \partial_t f_t=\frac12\Delta f_t.
\end{equation}
If the random variable $X$ has density $f$, then $P_t f$ is the density of $X+\sqrt t Z$, where $Z$ is a standard Gaussian vector independent of $X$.  For a smooth positive density $f_t$, the Fisher information is
\begin{equation}\label{eq:Fisher}
        I(t)=\int_{\R^d}\frac{|\nabla f_t|^2}{f_t}\dd x
             =\int_{\R^d}|\nabla \log f_t|^2 f_t\dd x.
\end{equation}
In 2015, Cheng and Geng \cite{ChengGeng2015} proved sign alternation for the third and fourth derivatives of entropy along one-dimensional heat flow and formulated two conjectures.  The first is the Gaussian completely monotone conjecture (GCMC), asserting complete monotonicity of the Fisher information $J(X+\sqrt t Z)$ along heat flow; the second is the log-convexity conjecture for this same Fisher-information profile.  In our notation, the second conjecture is
\begin{equation}\label{eq:log-convexity}
        I(t)I''(t)-I'(t)^2\ge0.
\end{equation}
The first conjecture implies the second, since every positive completely monotone function is log-convex.  There is also a broader hierarchy.  McKean's Gaussian optimality conjecture \cite{McKean1966} is stronger than GCMC, and Wang proved that Toscani's entropy power conjecture \cite{Toscani2015} implies McKean's conjecture \cite{Wang2024}.  Thus the implication chain takes the form
\begin{equation}\label{eq:hierarchy}
        \begin{gathered}
                \text{Entropy power conjecture}\\
                \Downarrow\\[-0.15em]
                \text{McKean's conjecture}\\
                \Downarrow\\[-0.15em]
                \text{GCMC}\\
                \Downarrow\\[-0.15em]
                \text{Log-convexity conjecture}.
        \end{gathered}
\end{equation}
A convenient static reformulation of the log-convexity defect is obtained by writing $f=e^u$ and $\Hmat_f=-\nabla^2u$, and then setting
\[
        \Ical[f]=\int_{\R^d}\tr(\Hmat_f)f\dd x,\qquad
        \Qcal[f]=\int_{\R^d}\tr(\Hmat_f^2)f\dd x,
\]
\[
        \Dcal[f]=\int_{\R^d}\Big(|\nabla\Hmat_f|^2+2\tr(\Hmat_f^3)\Big)f\dd x.
\]
Section \ref{sec:heat-flow} records the standard identities
\[
        I(t)=\Ical[f_t],\qquad I'(t)=-\Qcal[f_t],\qquad I''(t)=\Dcal[f_t],
\]
so that
\begin{equation}\label{eq:intro-static-defect}
        I(t)I''(t)-I'(t)^2=\Ical[f_t]\Dcal[f_t]-\Qcal[f_t]^2.
\end{equation}

There has been substantial work on this circle of conjectures.  Cheng and Geng established the first nontrivial higher-order one-dimensional sign results for the GCMC \cite{ChengGeng2015}; Zhang, Anantharam, and Geng, together with later semidefinite-programming work, developed effective tools for log-concave and related higher-order inequalities \cite{ZhangAnantharamGeng2018,GuoYuanGao2021}; and Liu and Gao proved the weaker square-root convexity inequality in dimension two \cite{LiuGao2023}, namely
\begin{equation}\label{eq:liu-gao}
        2\Ical[f]\Dcal[f]\ge \Qcal[f]^2.
\end{equation}
In dimension one, Ledoux, Nair, and Wang proved \eqref{eq:log-convexity} \cite{LedouxNairWang}.  More recently, Gu and Sellke disproved GCMC in dimension one and therefore also disproved the one-dimensional McKean and entropy power conjectures \cite{GuSellke2026}.  The present paper supplies the complementary multidimensional disproof.  Thus the conjectural hierarchy \eqref{eq:hierarchy} is now false in every dimension, while Fisher-information log-convexity survives only in dimension one.  All quantities in this paper are computed with the convention \eqref{eq:heat}; replacing it by $\partial_t f_t=\Delta f_t$ only rescales time and does not affect the sign of \eqref{eq:log-convexity}.

Our main result shows that the stronger inequality \eqref{eq:log-convexity} is false already in dimension two.

\begin{theorem}[A two-dimensional Euclidean counterexample]\label{thm:main}
There exists a smooth, strictly positive, Gaussian-decaying probability density $f$ on $\R^2$ such that
\begin{equation}\label{eq:counter-ineq}
        \Ical[f]\Dcal[f]-\Qcal[f]^2<0.
\end{equation}
If $f_t=P_t f$ and
\[
        \Phi_f(t):=I(t)I''(t)-I'(t)^2,
\]
then $\Phi_f(0)<0$.  Since $\Phi_f$ is continuous at $t=0$, it follows that
\[
        I(t)I''(t)-I'(t)^2<0
\]
for all sufficiently small $t>0$.  In particular, Fisher information along the heat flow is not log-convex in dimension two.
\end{theorem}

By tensorization with broad Gaussians we also obtain the higher-dimensional failure.

\begin{corollary}[Higher-dimensional failure by tensorization]\label{cor:high-d}
For every $d\ge2$, there exists a smooth, strictly positive, Gaussian-decaying probability density $f$ on $\R^d$ such that
\[
        \Ical[f]\Dcal[f]-\Qcal[f]^2<0.
\]
If $f_t=P_t f$, then the same continuity argument gives
\[
        I(t)I''(t)-I'(t)^2<0
\]
for all sufficiently small $t>0$.  Hence the corresponding heat flow violates \eqref{eq:log-convexity}.
\end{corollary}

The implication chain \eqref{eq:hierarchy} then gives the following immediate consequence.

\begin{corollary}[Failure of GCMC, McKean, and the entropy power conjecture]\label{cor:hierarchy-false}
For every $d\ge2$, the multidimensional Gaussian completely monotone conjecture, McKean's Gaussian optimality conjecture, and Toscani's entropy power conjecture are false.
\end{corollary}

This is immediate from Corollary \ref{cor:high-d} and the implication chain \eqref{eq:hierarchy}.  Combined with the one-dimensional counterexample of Gu and Sellke \cite{GuSellke2026}, Corollary \ref{cor:hierarchy-false} shows that the Gaussian completely monotone, McKean, and entropy power conjectures fail in every dimension.  Combined with the one-dimensional theorem of Ledoux, Nair, and Wang \cite{LedouxNairWang}, Theorem \ref{thm:main} shows that Fisher-information log-convexity holds exactly in dimension one.

The construction is explicit and perturbative.  On the triangular torus we work with a closed frequency triad of unit vectors
\[
        k_1+k_2+k_3=0,
        \qquad
        |k_1|=|k_2|=|k_3|=1.
\]
We call this zero-sum frequency relation a cubic frequency resonance: it is precisely the condition that a product of three Fourier modes can contain a zero Fourier mode.  A single unit-frequency mode is neutral at quadratic order, whereas the hexagonal three-wave interaction produces a nonzero cubic correction with the sign needed to violate \eqref{eq:log-convexity}.  We then place a Gaussian envelope on the torus profile to transfer this negative sign to a smooth Gaussian-decaying density on $\R^2$.

This does not contradict the square-root convexity inequality of Liu and Gao \eqref{eq:liu-gao}.  For the perturbative family constructed below,
\[
        \frac{\Ical[f]\Dcal[f]}{\Qcal[f]^2}=1-\frac18\eps+O(\eps^2),
\]
so the ratio remains strictly larger than $1/2$ for sufficiently small $\eps>0$.  This naturally leads to the sharp constant
\begin{equation}\label{eq:theta-star}
        \theta_d^*=\inf_{f\in\mathcal A_d}\frac{\Ical[f]\Dcal[f]}{\Qcal[f]^2}
\end{equation}
for $d\ge1$, where the admissible class $\mathcal A_d$ is specified in Section \ref{sec:open-problems}.  The present result shows $\theta_d^*<1$ for all $d\ge2$, while \cite{LiuGao2023} gives $\theta_2^*\ge1/2$.

Section \ref{sec:heat-flow} records the heat-flow identities behind the static formulation.  Sections \ref{sec:torus} and \ref{sec:transfer} prove Theorem \ref{thm:main} by analyzing the hexagonal torus perturbation and then transferring it to $\R^2$.  Section \ref{sec:numerics} records explicit numerical values for the two-dimensional family.  Section \ref{sec:high-d} gives two higher-dimensional extensions, one by tensorization and one by a frequency-lattice generalization of the hexagonal zero-mode mechanism in $\R^d$.  Section \ref{sec:open-problems} turns to the sharp constants $\theta_d^*$ and to the dichotomy tied to Fisher convexity, and the brief concluding section summarizes the remaining directions.

\section{Heat-flow Identities}\label{sec:heat-flow}

For a smooth positive probability density $f=e^u$ on $\R^d$ with sufficient decay, define
\begin{equation}\label{eq:Hmat}
        \Hmat_f=-\nabla^2u.
\end{equation}
We use the following descriptive notation:
\begin{equation}\label{eq:I-def}
        \Ical[f]=\int_{\R^d}\tr(\Hmat_f)f\dd x,
\end{equation}
\begin{equation}\label{eq:Q-def}
        \Qcal[f]=\int_{\R^d}\tr(\Hmat_f^2)f\dd x,
\end{equation}
\begin{equation}\label{eq:D-def}
        \Dcal[f]=\int_{\R^d}\Big(|\nabla\Hmat_f|^2+2\tr(\Hmat_f^3)\Big)f\dd x.
\end{equation}
Here
\[
        |\nabla\Hmat_f|^2=\sum_{i,j,k}(\partial_k(\Hmat_f)_{ij})^2.
\]
The letter $\Ical$ is used because $\Ical[f]$ is exactly Fisher information.  The symbols $\Qcal$ and $\Dcal$ denote, respectively, the quadratic logarithmic-Hessian energy and the second Fisher-information production functional.

\begin{lemma}[Heat-flow identities for Fisher information]\label{lem:identities}
Let $f_t$ solve \eqref{eq:heat}.  Then
\begin{equation}\label{eq:I=Ical}
        I(t)=\Ical[f_t],
\end{equation}
\begin{equation}\label{eq:Iprime}
        I'(t)=-\Qcal[f_t],
\end{equation}
and
\begin{equation}\label{eq:Isecond}
        I''(t)=\Dcal[f_t].
\end{equation}
Consequently,
\begin{equation}\label{eq:log-conv-equivalence}
        I(t)I''(t)-I'(t)^2
        =\Ical[f_t]\Dcal[f_t]-\Qcal[f_t]^2.
\end{equation}
\end{lemma}

\begin{proof}
First,
\[
        \Ical[f]=\int -\Delta u\,e^u\dd x
        =\int |\nabla u|^2e^u\dd x=I[f]
\]
by integration by parts.  Under \eqref{eq:heat}, the logarithmic density satisfies
\begin{equation}\label{eq:u-evol}
        \partial_t u=\frac12\Delta u+\frac12|\nabla u|^2.
\end{equation}
The standard Fisher-information dissipation identity gives
\[
        \frac{\dd}{\dd t}I(t)=-\int\tr(\Hmat_{f_t}^2)f_t\dd x=-\Qcal[f_t].
\]
Differentiating once more, equivalently applying the Bochner identity to \eqref{eq:u-evol}, gives
\[
        \frac{\dd^2}{\dd t^2}I(t)
        =\int\Big(|\nabla\Hmat_{f_t}|^2+2\tr(\Hmat_{f_t}^3)\Big)f_t\dd x
        =\Dcal[f_t].
\]
This proves \eqref{eq:log-conv-equivalence}.
\end{proof}

\begin{remark}[Other heat conventions]
If one uses $\partial_s f_s=\Delta f_s$, then $s=t/2$ relative to \eqref{eq:heat}.  The identities become $\frac{\dd}{\dd s}I=-2\Qcal$ and $\frac{\dd^2}{\dd s^2}I=4\Dcal$.  The sign of $II''-(I')^2$ is unchanged.  This paper uses \eqref{eq:heat} everywhere to match the convention $X_t=X+\sqrt t Z$.
\end{remark}

\section{The Hexagonal Perturbation on the Triangular Torus}\label{sec:torus}

Consider the three unit vectors
\begin{equation}\label{eq:k-vectors}
        k_1=(1,0),\qquad
        k_2=\left(-\frac12,\frac{\sqrt3}{2}\right),\qquad
        k_3=\left(-\frac12,-\frac{\sqrt3}{2}\right).
\end{equation}
They satisfy
\begin{equation}\label{eq:k-relations}
        |k_1|=|k_2|=|k_3|=1,
        \qquad k_1+k_2+k_3=0.
\end{equation}
Introduce the dual vectors
\begin{equation}\label{eq:l-vectors}
        \ell_1=2\pi\left(1,\frac1{\sqrt3}\right),
        \qquad
        \ell_2=2\pi\left(0,\frac2{\sqrt3}\right).
\end{equation}
which satisfy
\begin{equation}\label{eq:dual-relations}
        k_i\cdot \ell_j=2\pi\delta_{ij},
        \qquad i,j\in\{1,2\}.
\end{equation}
Accordingly, define the lattice
\begin{equation}\label{eq:lattice-def}
        \Lambda
        =
        \Z \ell_1\oplus \Z \ell_2
        =
        \{x\in\R^2:k_1\cdot x\in2\pi\Z,\ k_2\cdot x\in2\pi\Z\}.
\end{equation}
This is the physical-space lattice dual to the frequency basis $(k_1,k_2)$.  By \eqref{eq:dual-relations}, the phases $k_1\cdot x$ and $k_2\cdot x$ are $2\pi$-periodic on the quotient, and the same is therefore true of $k_3\cdot x$ because $k_3=-k_1-k_2$.  Write
\[
        \T^2_\Lambda=\R^2/\Lambda,
        \qquad
        \Omega=\{s\ell_1+t\ell_2:0\le s,t<1\}.
\]
The set $\Omega$ is a fundamental parallelogram for $\Lambda$, with area
\[
        |\Omega|=|\det(\ell_1,\ell_2)|=\frac{8\pi^2}{\sqrt3}.
\]
For every $\Lambda$-periodic integrable function $G$, we write
\begin{equation}\label{eq:haar-average-def}
        \av{G}
        =
        \frac{1}{|\Omega|}\int_\Omega G(x)\dd x.
\end{equation}
Thus $\av{G}$ is the normalized Haar average on $\T^2_\Lambda$, equivalently the average of $G$ over any fundamental cell of $\Lambda$.  Figure \ref{fig:lattice} records the geometry of $\Omega$ and the lattice points.

\begin{figure}[H]
\centering
\begin{tikzpicture}[scale=0.5]
\coordinate (O) at (0,0);
\coordinate (Lone) at (2.8,1.6166);
\coordinate (Ltwo) at (0,3.2332);
\coordinate (Lsum) at (2.8,4.8498);
\fill[blue!8] (O)--(Lone)--(Lsum)--(Ltwo)--cycle;
\foreach \m in {-1,0,1,2}{
    \foreach \n in {-1,0,1,2}{
        \fill (\m*2.8,\m*1.6166+\n*3.2332) circle (1.4pt);
    }
}
\draw[thin, gray!60] (-2.8,-4.8498) -- (5.6,-0.0000);
\draw[thin, gray!60] (-2.8,-1.6166) -- (5.6,3.2332);
\draw[thin, gray!60] (-2.8,1.6166) -- (5.6,6.4664);
\draw[thin, gray!60] (0,-3.2332) -- (0,6.4664);
\draw[thin, gray!60] (2.8,-1.6166) -- (2.8,8.0830);
\draw[thick] (O)--(Lone)--(Lsum)--(Ltwo)--cycle;
\draw[->, thick] (O)--(Lone);
\draw[->, thick] (O)--(Ltwo);
\fill (O) circle (1.6pt);
\node[below left] at (O) {$0$};
\node[below right] at (1.55,0.65) {$\ell_1$};
\node[left] at (-0.08,1.7) {$\ell_2$};
\node at (1.15,2.2) {$\Omega$};
\end{tikzpicture}
\caption{A fundamental parallelogram $\Omega$ for the triangular lattice $\Lambda=\Z\ell_1\oplus\Z\ell_2$.  The torus $\T^2_\Lambda$ is obtained by identifying opposite sides.}
\label{fig:lattice}
\end{figure}
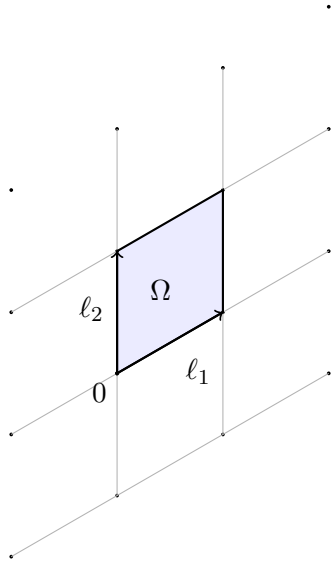

Set
\begin{equation}\label{eq:phi-def}
        \theta_j(x)=k_j\cdot x,
        \qquad
        \phi(x)=\cos\theta_1+\cos\theta_2+\cos\theta_3.
\end{equation}
For $\eps\in\R$, define
\begin{equation}\label{eq:feps-torus}
        f_\eps=Z_\eps^{-1}e^{\eps\phi},
        \qquad
        Z_\eps=\av{e^{\eps\phi}}.
\end{equation}
Its logarithm is
\[
        u_\eps=\log f_\eps=\eps\phi-\log Z_\eps,
\]
and hence
\begin{equation}\label{eq:Hmat-eps}
        \Hmat_{f_\eps}=-\eps\nabla^2\phi.
\end{equation}

\begin{lemma}[Hexagonal averages]\label{lem:averages}
For $\phi$ defined by \eqref{eq:phi-def},
\begin{align}
        \av{|\nabla\phi|^2}&=\frac32,&
        \av{\phi |\nabla\phi|^2}&=\frac34,\label{eq:avg1}\\
        \av{|\nabla^2\phi|^2}&=\frac32,&
        \av{\phi |\nabla^2\phi|^2}&=\frac38,\label{eq:avg2}\\
        \av{|\nabla\nabla^2\phi|^2}&=\frac32,&
        \av{\phi |\nabla\nabla^2\phi|^2}&=\frac3{16},\label{eq:avg3}\\
        \av{\tr((\nabla^2\phi)^3)}&=\frac3{16}.\label{eq:avg4}
\end{align}
\end{lemma}

\begin{proof}
Write
\[
        x=\frac{s}{2\pi}\ell_1+\frac{t}{2\pi}\ell_2,
        \qquad 0\le s,t<2\pi.
\]
By \eqref{eq:dual-relations}, this gives
\[
        \theta_1(x)=s,
        \qquad
        \theta_2(x)=t,
        \qquad
        \theta_3(x)=-(s+t).
\]
Since $\dd x=\frac{|\Omega|}{4\pi^2}\dd s\,\dd t$, the normalized average over $\Omega$ becomes
\begin{equation}\label{eq:torus-average-st}
        \av{F(\theta_1,\theta_2)}
        =
        \frac1{4\pi^2}\int_0^{2\pi}\int_0^{2\pi}F(s,t)\dd s\,\dd t.
\end{equation}
In particular, for every $(m,n)\in\Z^2$,
\begin{equation}\label{eq:torus-orthogonality}
        \av{e^{i(m\theta_1+n\theta_2)}}=
        \begin{cases}
                1,&(m,n)=(0,0),\\
                0,&(m,n)\ne(0,0).
        \end{cases}
\end{equation}
Thus every nonconstant trigonometric mode in $(\theta_1,\theta_2)$ has zero average.

We first compute the quadratic quantities.  Since
\[
        \nabla\phi=-\sum_{j=1}^3\sin\theta_j\,k_j,
\]
we obtain
\[
        |\nabla\phi|^2
        =
        \sum_{j=1}^3\sin^2\theta_j\,|k_j|^2
        +
        2\sum_{1\le i<j\le3}(k_i\cdot k_j)\sin\theta_i\sin\theta_j.
\]
Here $|k_j|=1$ and, for $i\ne j$, $k_i\cdot k_j=-\frac12$.  Moreover,
\[
        \av{\sin^2\theta_j}=\frac12,
\]
and for $i\ne j$,
\[
        \sin\theta_i\sin\theta_j
        =
        \frac12\bigl(\cos(\theta_i-\theta_j)-\cos(\theta_i+\theta_j)\bigr)
\]
has zero average by \eqref{eq:torus-orthogonality}.  Therefore
\[
        \av{|\nabla\phi|^2}
        =
        \sum_{j=1}^3\av{\sin^2\theta_j}
        =
        \frac32.
\]

Next set $M_j=k_jk_j^T$.  Then
\[
        \nabla^2\phi=-\sum_{j=1}^3\cos\theta_j\,M_j.
\]
Since $M_j:M_j=\tr(M_j^2)=|k_j|^4=1$, and for $i\ne j$,
\[
        M_i:M_j
        =
        \tr(k_ik_i^Tk_jk_j^T)
        =
        (k_i\cdot k_j)^2
        =
        \frac14,
\]
we have
\[
        |\nabla^2\phi|^2
        =
        \sum_{j=1}^3\cos^2\theta_j
        +
        2\sum_{1\le i<j\le3}(M_i:M_j)\cos\theta_i\cos\theta_j.
\]
Again \(\av{\cos^2\theta_j}=\frac12\), while for $i\ne j$,
\[
        \cos\theta_i\cos\theta_j
        =
        \frac12\bigl(\cos(\theta_i-\theta_j)+\cos(\theta_i+\theta_j)\bigr)
\]
has zero average.  Hence
\[
        \av{|\nabla^2\phi|^2}=\frac32.
\]

For the third derivative tensor,
\[
        \nabla\nabla^2\phi
        =
        \sum_{j=1}^3\sin\theta_j\,k_j^{\otimes3}.
\]
Now
\[
        k_j^{\otimes3}:k_j^{\otimes3}=|k_j|^6=1,
\]
and for $i\ne j$,
\[
        k_i^{\otimes3}:k_j^{\otimes3}
        =
        (k_i\cdot k_j)^3
        =
        -\frac18.
\]
Therefore
\[
        |\nabla\nabla^2\phi|^2
        =
        \sum_{j=1}^3\sin^2\theta_j
        +
        2\sum_{1\le i<j\le3}
        \bigl(k_i^{\otimes3}:k_j^{\otimes3}\bigr)\sin\theta_i\sin\theta_j,
\]
and the same orthogonality argument shows that all cross terms average to zero.  Thus
\[
        \av{|\nabla\nabla^2\phi|^2}=\frac32.
\]

We next turn to the cubic averages.  Since $\theta_3=-(\theta_1+\theta_2)$, we have
\[
        \cos\theta_1\cos\theta_2\cos\theta_3
        =
        \cos s\cos t\cos(s+t).
\]
Using $\cos s\cos t=\frac12(\cos(s+t)+\cos(s-t))$, we obtain
\[
        \cos s\cos t\cos(s+t)
        =
        \frac14\bigl(1+\cos2s+\cos2t+\cos2(s+t)\bigr).
\]
By \eqref{eq:torus-average-st}, all nonconstant terms average to zero, so
\begin{equation}\label{eq:cos123-average}
        \av{\cos\theta_1\cos\theta_2\cos\theta_3}=\frac14.
\end{equation}

Expanding $\phi^3=(\cos\theta_1+\cos\theta_2+\cos\theta_3)^3$, we find
\[
        \phi^3
        =
        \sum_{j=1}^3\cos^3\theta_j
        +
        3\sum_{i\ne j}\cos^2\theta_i\cos\theta_j
        +
        6\cos\theta_1\cos\theta_2\cos\theta_3.
\]
The first two groups have zero average: $\cos^3\theta_j$ contains only the modes $\cos\theta_j$ and $\cos3\theta_j$, while
\[
        \cos^2\theta_i\cos\theta_j
        =
        \frac12\cos\theta_j+\frac12\cos2\theta_i\cos\theta_j
\]
also has only nonconstant Fourier modes.  Hence, by \eqref{eq:cos123-average},
\begin{equation}\label{eq:phi3}
        \av{\phi^3}
        =
        6\av{\cos\theta_1\cos\theta_2\cos\theta_3}
        =
        \frac32.
\end{equation}

To compute $\av{\phi|\nabla\phi|^2}$, we use periodic integration by parts on $\T^2_\Lambda$:
\[
        \av{\phi|\nabla\phi|^2}
        =
        \frac12\av{\nabla(\phi^2)\cdot\nabla\phi}
        =
        -\frac12\av{\phi^2\Delta\phi}.
\]
Since each mode in $\phi$ has frequency norm $1$, we have $-\Delta\phi=\phi$.  Therefore
\[
        \av{\phi|\nabla\phi|^2}
        =
        \frac12\av{\phi^3}
        =
        \frac34.
\]

For $\av{\phi|\nabla^2\phi|^2}$, insert the expansion of $|\nabla^2\phi|^2$:
\[
        \phi|\nabla^2\phi|^2
        =
        \sum_{\ell=1}^3\sum_{j=1}^3
        \cos\theta_\ell\cos^2\theta_j
        +
        2\sum_{\ell=1}^3\sum_{1\le i<j\le3}
        (M_i:M_j)\cos\theta_\ell\cos\theta_i\cos\theta_j.
\]
Every term in the first double sum has zero average, by the same argument used for $\phi^3$.  In the second sum, if $\ell=i$ or $\ell=j$, we again get a term of the form $\cos^2\theta_i\cos\theta_j$, whose average is zero.  Thus only the terms with $\ell,i,j$ pairwise distinct survive.  There are exactly six such ordered contributions, each equal to
\[
        (M_i:M_j)\av{\cos\theta_1\cos\theta_2\cos\theta_3}
        =
        \frac14\cdot\frac14.
\]
Hence
\[
        \av{\phi|\nabla^2\phi|^2}
        =
        6\cdot\frac14\cdot\frac14
        =
        \frac38.
\]

For $\av{\phi|\nabla\nabla^2\phi|^2}$, write
\[
        \phi|\nabla\nabla^2\phi|^2
        =
        \sum_{\ell=1}^3\sum_{j=1}^3
        \cos\theta_\ell\sin^2\theta_j
        +
        2\sum_{\ell=1}^3\sum_{1\le i<j\le3}
        \bigl(k_i^{\otimes3}:k_j^{\otimes3}\bigr)
        \cos\theta_\ell\sin\theta_i\sin\theta_j.
\]
As before, the first double sum has zero average.  In the second sum, the cases $\ell=i$ or $\ell=j$ also vanish by orthogonality, so only pairwise distinct indices remain.  For such a triple, say $(i,j,\ell)=(1,2,3)$, we use $\theta_3=-(\theta_1+\theta_2)$ and obtain
\[
        \cos\theta_3\sin\theta_1\sin\theta_2
        =
        \cos(\theta_1+\theta_2)\sin\theta_1\sin\theta_2
        =
        \frac14\bigl(\cos2\theta_1+\cos2\theta_2-1-\cos2(\theta_1+\theta_2)\bigr).
\]
Therefore
\[
        \av{\cos\theta_\ell\sin\theta_i\sin\theta_j}=-\frac14
\]
whenever $i,j,\ell$ are distinct.  Since $k_i^{\otimes3}:k_j^{\otimes3}=-\frac18$ for $i\ne j$, it follows that
\[
        \av{\phi|\nabla\nabla^2\phi|^2}
        =
        6\left(-\frac18\right)\left(-\frac14\right)
        =
        \frac3{16}.
\]

Finally,
\[
        \tr\bigl((\nabla^2\phi)^3\bigr)
        =
        -\sum_{i,j,\ell=1}^3
        \cos\theta_i\cos\theta_j\cos\theta_\ell\,\tr(M_iM_jM_\ell).
\]
Since
\[
        M_iM_j=(k_i\cdot k_j)\,k_ik_j^T,
\]
we get
\[
        \tr(M_iM_jM_\ell)
        =
        (k_i\cdot k_j)(k_j\cdot k_\ell)(k_\ell\cdot k_i).
\]
If $i=j=\ell$, then $\tr(M_i^3)=1$ but $\av{\cos^3\theta_i}=0$.  If exactly two indices coincide, say $i=j\ne\ell$, then $\tr(M_i^2M_\ell)=(k_i\cdot k_\ell)^2=\frac14$, but $\av{\cos^2\theta_i\cos\theta_\ell}=0$.  Thus again only the six terms with pairwise distinct indices contribute.  For each of them,
\[
        \tr(M_iM_jM_\ell)
        =
        \left(-\frac12\right)^3
        =
        -\frac18,
\]
and \eqref{eq:cos123-average} gives the average factor $\frac14$.  Hence
\[
        \av{\tr((\nabla^2\phi)^3)}
        =
        -6\left(-\frac18\right)\left(\frac14\right)
        =
        \frac3{16}.
\]
\end{proof}

\begin{proposition}[Torus counterexample]\label{prop:torus}
For $f_\eps$ defined by \eqref{eq:feps-torus},
\begin{align}
        \Ical[f_\eps]
        &=\frac32\eps^2+\frac34\eps^3+O(\eps^4),\label{eq:Iexp}\\
        \Qcal[f_\eps]
        &=\frac32\eps^2+\frac38\eps^3+O(\eps^4),\label{eq:Qexp}\\
        \Dcal[f_\eps]
        &=\frac32\eps^2-\frac3{16}\eps^3+O(\eps^4).\label{eq:Dexp}
\end{align}
Consequently,
\begin{equation}\label{eq:torus-product-negative}
        \Ical[f_\eps]\Dcal[f_\eps]-\Qcal[f_\eps]^2
        =-\frac9{32}\eps^5+O(\eps^6),
\end{equation}
and
\begin{equation}\label{eq:quotient-torus}
        \frac{\Ical[f_\eps]\Dcal[f_\eps]}{\Qcal[f_\eps]^2}
        =1-\frac18\eps+O(\eps^2).
\end{equation}
Thus the torus form of the log-convexity inequality fails for all sufficiently small $\eps>0$.
\end{proposition}

\begin{proof}
Since $\tr(\Hmat_{f_\eps})=-\eps\Delta\phi=\eps\phi$, we have
\[
        \Ical[f_\eps]
        =\eps\frac{\av{\phi e^{\eps\phi}}}{\av{e^{\eps\phi}}}.
\]
Since $\av{\phi}=0$ and $\av{e^{\eps\phi}}=1+O(\eps^2)$, using \eqref{eq:phi3},
\[
        \Ical[f_\eps]
        =\eps\left(\eps\av{\phi^2}+\frac{\eps^2}{2}\av{\phi^3}\right)+O(\eps^4)
        =\frac32\eps^2+\frac34\eps^3+O(\eps^4).
\]

Next,
\[
        \tr(\Hmat_{f_\eps}^2)=\eps^2|\nabla^2\phi|^2,
\]
so Lemma \ref{lem:averages} gives
\[
        \Qcal[f_\eps]
        =\eps^2\frac{\av{|\nabla^2\phi|^2e^{\eps\phi}}}{\av{e^{\eps\phi}}}
        =\frac32\eps^2+\frac38\eps^3+O(\eps^4).
\]

Finally,
\[
        |\nabla\Hmat_{f_\eps}|^2=\eps^2|\nabla\nabla^2\phi|^2,
\]
and
\[
        \tr(\Hmat_{f_\eps}^3)=-\eps^3\tr((\nabla^2\phi)^3).
\]
Therefore
\[
\begin{aligned}
        \Dcal[f_\eps]
        & = \eps^2\av{|\nabla\nabla^2\phi|^2}
        +\eps^3\left(\av{\phi|\nabla\nabla^2\phi|^2}
        -2\av{\tr((\nabla^2\phi)^3)}\right)
        +O(\eps^4) \\
        & =\frac32\eps^2+\left(\frac3{16}-2\cdot\frac3{16}\right)\eps^3+O(\eps^4)\\
        & =\frac32\eps^2-\frac3{16}\eps^3+O(\eps^4).
\end{aligned}
\]
Write
\[
\Ical=a_2\eps^2+a_3\eps^3+O(\eps^4),\quad
\Qcal=q_2\eps^2+q_3\eps^3+O(\eps^4),\quad
\Dcal=d_2\eps^2+d_3\eps^3+O(\eps^4).
\]
Here
\[
        a_2=q_2=d_2=\frac32,
        \qquad
        a_3=\frac34,
        \qquad
        q_3=\frac38,
        \qquad
        d_3=-\frac3{16}.
\]
Thus the coefficient of $\eps^5$ in $\Ical\Dcal-\Qcal^2$ is
\[
        a_2d_3+a_3d_2-2q_2q_3
        =\frac32\left(-\frac3{16}\right)+\frac34\cdot\frac32-2\cdot\frac32\cdot\frac38
        =-\frac9{32}.
\]
This proves \eqref{eq:torus-product-negative}.  Dividing by $\Qcal[f_\eps]^2$ gives \eqref{eq:quotient-torus}.
\end{proof}

\section{Transfer to the Euclidean Space}\label{sec:transfer}

We next transfer the torus calculation to a smooth rapidly decaying counterexample on $\R^2$.
For $R>0$, define
\begin{equation}\label{eq:FR-def}
        F_{R,\eps}(x)=Z_{R,\eps}^{-1}\exp\left(\eps\phi(x)-\frac{|x|^2}{2R^2}\right),
        \qquad x\in\R^2.
\end{equation}
Then $F_{R,\eps}$ is smooth, strictly positive, and Gaussian-decaying.  Its logarithmic Hessian is
\begin{equation}\label{eq:HR-eps}
        \Hmat_{F_{R,\eps}}
        =-\eps\nabla^2\phi+R^{-2}I_2.
\end{equation}
Also
\begin{equation}\label{eq:gradHR}
        \nabla\Hmat_{F_{R,\eps}}
        =-\eps\nabla\nabla^2\phi,
\end{equation}
since the Gaussian Hessian shift $R^{-2}I_2$ is constant.

\begin{lemma}[Gaussian-envelope averaging]\label{lem:Gaussian-envelope}
Let $G$ be continuous and $\Lambda$-periodic.  Then, for each fixed $\eps$,
\begin{equation}\label{eq:Gaussian-average}
        \frac{\int_{\R^2}G(x)e^{\eps\phi(x)}e^{-|x|^2/(2R^2)}\dd x}
        {\int_{\R^2}e^{\eps\phi(x)}e^{-|x|^2/(2R^2)}\dd x}
        \longrightarrow
        \frac{\av{Ge^{\eps\phi}}}{\av{e^{\eps\phi}}}
        \qquad (R\to\infty).
\end{equation}
\end{lemma}

\begin{proof}
For a continuous $\Lambda$-periodic function $H$, define
\[
        \mathcal G_R(H)
        :=
        \int_{\R^2}H(x)e^{-|x|^2/(2R^2)}\dd x.
\]
We first prove that
\begin{equation}\label{eq:gaussian-claim}
        \mathcal G_R(H)=2\pi R^2\av{H}+o(R^2)
        \qquad(R\to\infty)
\end{equation}
for every such $H$.  Once this is proved, the lemma follows immediately by applying \eqref{eq:gaussian-claim} to
\[
        H_1=Ge^{\eps\phi}
        \qquad\text{and}\qquad
        H_0=e^{\eps\phi},
\]
because $\av{H_0}=\av{e^{\eps\phi}}>0$.

To prove \eqref{eq:gaussian-claim}, fix $\eta>0$.  The trigonometric polynomials in
\[
        e^{i\theta_1(x)}
        \qquad\text{and}\qquad
        e^{i\theta_2(x)}
\]
are uniformly dense in $C(\T^2_\Lambda)$, so there exists a finite set
$\mathcal F_\eta\subset\Z^2$ and a finite Fourier sum
\[
        T_\eta(x)=\sum_{(m,n)\in\mathcal F_\eta}c_{m,n}e^{i(m\theta_1(x)+n\theta_2(x))}
\]
such that
\[
        \|H-T_\eta\|_{L^\infty(\T^2_\Lambda)}\le\eta.
\]
For each pair $(m,n)\in\Z^2$, we have
\[
        e^{i(m\theta_1(x)+n\theta_2(x))}
        =
        e^{i(mk_1+n k_2)\cdot x},
\]
and therefore the Gaussian Fourier transform gives
\[
        \mathcal G_R\!\left(e^{i(m\theta_1+n\theta_2)}\right)
        =
        2\pi R^2
        \exp\left(-\frac{R^2}{2}|mk_1+n k_2|^2\right).
\]
Since
\[
        |mk_1+n k_2|^2
        =
        m^2+n^2-mn
        =
        \frac12\bigl(m^2+(m-n)^2+n^2\bigr),
\]
we have $|mk_1+n k_2|^2\ge1$ for every $(m,n)\ne(0,0)$.  Hence
\[
        \mathcal G_R(T_\eta)
        =
        2\pi R^2 c_{0,0}
        +
        O_{T_\eta}(R^2e^{-R^2/2}),
\]
where the implied constant depends only on the finitely many coefficients of $T_\eta$.  Because $c_{0,0}=\av{T_\eta}$, this can be rewritten as
\begin{equation}\label{eq:gaussian-poly}
        \mathcal G_R(T_\eta)
        =
        2\pi R^2\av{T_\eta}
        +
        O_{T_\eta}(R^2e^{-R^2/2}).
\end{equation}

Now compare $H$ with $T_\eta$.  Since $\int_{\R^2}e^{-|x|^2/(2R^2)}\dd x=2\pi R^2$, we have
\[
        |\mathcal G_R(H-T_\eta)|
        \le
        \|H-T_\eta\|_{L^\infty}\int_{\R^2}e^{-|x|^2/(2R^2)}\dd x
        \le
        2\pi R^2\eta.
\]
Also,
\[
        |\av{H}-\av{T_\eta}|
        \le
        \av{|H-T_\eta|}
        \le
        \eta.
\]
Combining these bounds with \eqref{eq:gaussian-poly}, we obtain
\[
        \begin{aligned}
        |\mathcal G_R(H)-2\pi R^2\av{H}|
        &\le
        |\mathcal G_R(H-T_\eta)|
        +
        |\mathcal G_R(T_\eta)-2\pi R^2\av{T_\eta}|
        +
        2\pi R^2|\av{T_\eta}-\av{H}|\\
        &\le
        4\pi R^2\eta
        +
        O_{T_\eta}(R^2e^{-R^2/2}).
        \end{aligned}
\]
Divide by $R^2$ and let $R\to\infty$ to get
\[
        \limsup_{R\to\infty}\frac{|\mathcal G_R(H)-2\pi R^2\av{H}|}{R^2}
        \le
        4\pi\eta.
\]
Since $\eta>0$ is arbitrary, \eqref{eq:gaussian-claim} follows.

Applying \eqref{eq:gaussian-claim} to $H_1$ and $H_0$ yields
\[
        \int_{\R^2}G(x)e^{\eps\phi(x)}e^{-|x|^2/(2R^2)}\dd x
        =
        2\pi R^2\av{Ge^{\eps\phi}}+o(R^2),
\]
\[
        \int_{\R^2}e^{\eps\phi(x)}e^{-|x|^2/(2R^2)}\dd x
        =
        2\pi R^2\av{e^{\eps\phi}}+o(R^2).
\]
Taking the ratio proves \eqref{eq:Gaussian-average}.
\end{proof}

\begin{proposition}[Euclidean counterexample]\label{prop:euclidean}
For each sufficiently small $\eps>0$, there exists $R_0(\eps)<\infty$ such that for all $R\ge R_0(\eps)$,
\begin{equation}\label{eq:Euclidean-negative}
        \Ical[F_{R,\eps}]\Dcal[F_{R,\eps}]-\Qcal[F_{R,\eps}]^2<0.
\end{equation}
\end{proposition}

\begin{proof}
Fix $\eps>0$.  For every continuous $\Lambda$-periodic function $G$, write
\[
        \mathbb E_R[G]
        :=
        \frac{\int_{\R^2}G(x)e^{\eps\phi(x)}e^{-|x|^2/(2R^2)}\dd x}
        {\int_{\R^2}e^{\eps\phi(x)}e^{-|x|^2/(2R^2)}\dd x},
\]
\[
        \mathbb E_\infty[G]
        :=
        \frac{\av{Ge^{\eps\phi}}}{\av{e^{\eps\phi}}}.
\]
Lemma \ref{lem:Gaussian-envelope} says exactly that
\begin{equation}\label{eq:ER-to-Einfty}
        \mathbb E_R[G]\longrightarrow\mathbb E_\infty[G]
        \qquad(R\to\infty)
\end{equation}
for every such $G$.

Set
\[
        A(x):=-\eps\nabla^2\phi(x),
        \qquad
        c_R:=R^{-2}.
\]
Then \eqref{eq:HR-eps} and \eqref{eq:gradHR} become
\[
        \Hmat_{F_{R,\eps}}=A+c_RI_2,
        \qquad
        \nabla\Hmat_{F_{R,\eps}}=\nabla A.
\]
Since $A$ and $I_2$ commute, we may expand the relevant polynomial expressions exactly.

First, because $\tr(A)=-\eps\Delta\phi=\eps\phi$, we have
\begin{equation}\label{eq:I-euclidean-expanded}
        \Ical[F_{R,\eps}]
        =
        \mathbb E_R[\tr(A+c_RI_2)]
        =
        \mathbb E_R[\eps\phi]+2c_R.
\end{equation}
On the torus side, \eqref{eq:feps-torus} and \eqref{eq:Hmat-eps} give
\begin{equation}\label{eq:I-torus-ER}
        \Ical[f_\eps]=\mathbb E_\infty[\eps\phi].
\end{equation}

Next,
\[
        \tr\bigl((A+c_RI_2)^2\bigr)
        =
        \tr(A^2)+2c_R\tr(A)+c_R^2\tr(I_2).
\]
Using $\tr(A^2)=\eps^2|\nabla^2\phi|^2$ and $\tr(I_2)=2$, we obtain
\begin{equation}\label{eq:Q-euclidean-expanded}
        \Qcal[F_{R,\eps}]
        =
        \mathbb E_R[\eps^2|\nabla^2\phi|^2]
        +
        2\eps c_R\mathbb E_R[\phi]
        +
        2c_R^2.
\end{equation}
Similarly,
\begin{equation}\label{eq:Q-torus-ER}
        \Qcal[f_\eps]=\mathbb E_\infty[\eps^2|\nabla^2\phi|^2].
\end{equation}

For the third functional, let
\[
        D_0
        :=
        \eps^2|\nabla\nabla^2\phi|^2
        -2\eps^3\tr\bigl((\nabla^2\phi)^3\bigr).
\]
Then $\Dcal[f_\eps]=\mathbb E_\infty[D_0]$.  Also
\[
        2\tr\bigl((A+c_RI_2)^3\bigr)
        =
        2\tr(A^3)
        +6c_R\tr(A^2)
        +6c_R^2\tr(A)
        +2c_R^3\tr(I_2),
\]
so, using $\tr(A^3)=-\eps^3\tr((\nabla^2\phi)^3)$ and $\tr(I_2)=2$,
\begin{equation}\label{eq:D-euclidean-expanded}
        \Dcal[F_{R,\eps}]
        =
        \mathbb E_R[D_0]
        +
        6\eps^2c_R\mathbb E_R[|\nabla^2\phi|^2]
        +
        6\eps c_R^2\mathbb E_R[\phi]
        +
        4c_R^3.
\end{equation}

The functions
\[
        \phi,\qquad
        |\nabla^2\phi|^2,\qquad
        |\nabla\nabla^2\phi|^2,\qquad
        \tr\bigl((\nabla^2\phi)^3\bigr),\qquad
        D_0
\]
are all smooth and $\Lambda$-periodic, hence bounded and continuous.  Therefore \eqref{eq:ER-to-Einfty} applies to each of them, and in particular
\[
        \mathbb E_R[\eps\phi]\to\mathbb E_\infty[\eps\phi]=\Ical[f_\eps],
\]
\[
        \mathbb E_R[\eps^2|\nabla^2\phi|^2]
        \to
        \mathbb E_\infty[\eps^2|\nabla^2\phi|^2]
        =
        \Qcal[f_\eps],
\]
\[
        \mathbb E_R[D_0]\to\mathbb E_\infty[D_0]=\Dcal[f_\eps].
\]
Moreover, boundedness gives
\[
        |\mathbb E_R[\phi]|\le\|\phi\|_{L^\infty},
        \qquad
        \mathbb E_R[|\nabla^2\phi|^2]\le\||\nabla^2\phi|^2\|_{L^\infty},
\]
so the explicit correction terms in \eqref{eq:I-euclidean-expanded}, \eqref{eq:Q-euclidean-expanded}, and \eqref{eq:D-euclidean-expanded} vanish as $R\to\infty$.  Consequently,
\begin{equation}\label{eq:IQD-transfer}
        \Ical[F_{R,\eps}]\to\Ical[f_\eps],\qquad
        \Qcal[F_{R,\eps}]\to\Qcal[f_\eps],\qquad
        \Dcal[F_{R,\eps}]\to\Dcal[f_\eps].
\end{equation}

Now define
\[
        \Delta_{R,\eps}
        :=
        \Ical[F_{R,\eps}]\Dcal[F_{R,\eps}]
        -
        \Qcal[F_{R,\eps}]^2,
\]
\[
        \Delta_\eps
        :=
        \Ical[f_\eps]\Dcal[f_\eps]
        -
        \Qcal[f_\eps]^2.
\]
Since the map $(a,q,d)\mapsto ad-q^2$ is continuous, \eqref{eq:IQD-transfer} implies
\[
        \Delta_{R,\eps}\longrightarrow\Delta_\eps
        \qquad(R\to\infty).
\]
By Proposition \ref{prop:torus}, we have $\Delta_\eps<0$ for all sufficiently small positive $\eps$.  Fix such an $\eps$, and set
\[
        \delta:=-\frac12\Delta_\eps>0.
\]
Then for all sufficiently large $R$,
\[
        |\Delta_{R,\eps}-\Delta_\eps|<\delta,
\]
which yields
\[
        \Delta_{R,\eps}
        <
        \Delta_\eps+\delta
        =
        \frac12\Delta_\eps
        <0.
\]
This is exactly \eqref{eq:Euclidean-negative}.
\end{proof}

We can now combine the torus sign computation with Proposition \ref{prop:euclidean} to complete the proof of the main theorem.

\begin{proof}[Proof of Theorem \ref{thm:main}]
Choose $\eps>0$ small enough so that the torus value in \eqref{eq:torus-product-negative} is negative.  Then choose $R$ large enough in Proposition \ref{prop:euclidean}.  The density $F_{R,\eps}$ is smooth, strictly positive, Gaussian-decaying, and satisfies \eqref{eq:counter-ineq}.

Let $F_{R,\eps,t}=e^{t\Delta/2}F_{R,\eps}$ and define
\[
        \Phi_{R,\eps}(t)
        :=\Ical[F_{R,\eps,t}]\Dcal[F_{R,\eps,t}]
          -\Qcal[F_{R,\eps,t}]^2.
\]
Then \eqref{eq:counter-ineq} says exactly that $\Phi_{R,\eps}(0)<0$.  Since $F_{R,\eps}$ is smooth and Gaussian-decaying, the quantities $\Ical[F_{R,\eps,t}]$, $\Qcal[F_{R,\eps,t}]$, and $\Dcal[F_{R,\eps,t}]$ depend continuously on $t$ near $0$, so $\Phi_{R,\eps}$ is continuous at $0$.  Therefore $\Phi_{R,\eps}(t)<0$ for all sufficiently small positive $t$.  By Lemma \ref{lem:identities},
\[
        \Phi_{R,\eps}(t)=I(t)I''(t)-I'(t)^2.
\]
Hence $(\log I(t))''<0$ for all sufficiently small positive $t$, and the heat flow is not log-convex.
\end{proof}

\section{Numerical Verification in Dimension Two}\label{sec:numerics}

Although the proof of Theorem \ref{thm:main} is analytic, it is useful to record explicit numerical values for the concrete counterexample family
\[
        F_{R,\eps}(x)=Z_{R,\eps}^{-1}\exp\left(\eps\phi(x)-\frac{|x|^2}{2R^2}\right),
        \qquad
        \phi(x)=\cos(k_1\cdot x)+\cos(k_2\cdot x)+\cos(k_3\cdot x).
\]
In Table \ref{tab:d2-numerics} we fix $R=1000$ and evaluate this family for several representative values of $\eps$.

To compute the three functionals, let
\[
        W_{R,\eps}(x)=e^{\eps\phi(x)}e^{-|x|^2/(2R^2)}.
\]
Since the normalization factor $Z_{R,\eps}^{-1}$ cancels in the ratios, we may write
\[
        \Ical[F_{R,\eps}]
        =
        \frac{\int_{\R^2}\tr(\Hmat_{F_{R,\eps}}(x))W_{R,\eps}(x)\dd x}
        {\int_{\R^2}W_{R,\eps}(x)\dd x},
\]
\[
        \Qcal[F_{R,\eps}]
        =
        \frac{\int_{\R^2}\tr(\Hmat_{F_{R,\eps}}(x)^2)W_{R,\eps}(x)\dd x}
        {\int_{\R^2}W_{R,\eps}(x)\dd x},
\]
\[
        \Dcal[F_{R,\eps}]
        =
        \frac{\int_{\R^2}\Bigl(|\nabla\Hmat_{F_{R,\eps}}(x)|^2
        +2\tr(\Hmat_{F_{R,\eps}}(x)^3)\Bigr)W_{R,\eps}(x)\dd x}
        {\int_{\R^2}W_{R,\eps}(x)\dd x}.
\]
All periodic integrands in these weighted averages depend on $x$ only through
\[
        s=k_1\cdot x,
        \qquad
        t=k_2\cdot x.
\]
Accordingly, for any one of the periodic functions being integrated, we write
\[
        G(s,t)=\sum_{m,n\in\mathbb Z}\widehat G(m,n)e^{i(ms+nt)},
\]
where
\[
        \widehat G(m,n)
        =
        \frac{1}{4\pi^2}\int_0^{2\pi}\int_0^{2\pi}
        G(s,t)e^{-i(ms+nt)}\dd s\,\dd t
\]
is the $(m,n)$-th Fourier coefficient.  Since the periodic functions used here are smooth, their Fourier series converges absolutely, so we may integrate term by term.  Writing
\[
        \xi_{m,n}:=mk_1+n k_2,
\]
we have
\[
        e^{i(ms+nt)}=e^{i(mk_1+n k_2)\cdot x}=e^{i\xi_{m,n}\cdot x}.
\]
Hence
\[
        \int_{\R^2}G(x)e^{-|x|^2/(2R^2)}\dd x
        =
        \sum_{m,n\in\mathbb Z}\widehat G(m,n)
        \int_{\R^2}e^{i\xi_{m,n}\cdot x}e^{-|x|^2/(2R^2)}\dd x.
\]
The inner integral is the Fourier transform of the centered Gaussian:
\[
        \int_{\R^2}e^{i\xi\cdot x}e^{-|x|^2/(2R^2)}\dd x
        =
        2\pi R^2 e^{-R^2|\xi|^2/2},
\]
so with $\xi=\xi_{m,n}$ we obtain
\[
        \int_{\R^2}e^{i\xi_{m,n}\cdot x}e^{-|x|^2/(2R^2)}\dd x
        =
        2\pi R^2 e^{-R^2|\xi_{m,n}|^2/2}.
\]
Finally, using
\[
        |k_1|=|k_2|=1,
        \qquad
        k_1\cdot k_2=-\frac12,
\]
we compute
\[
        |\xi_{m,n}|^2
        =
        |mk_1+n k_2|^2
        =
        m^2+n^2+2mn(k_1\cdot k_2)
        =
        m^2+n^2-mn.
\]
Therefore
\begin{equation}\label{eq:numerical-fourier}
        \int_{\R^2}G(x)e^{-|x|^2/(2R^2)}\dd x
        =2\pi R^2\sum_{m,n\in\mathbb Z}\widehat G(m,n)
        \exp\left(-\frac{R^2}{2}(m^2+n^2-mn)\right).
\end{equation}
We computed the relevant Fourier coefficients on a $256\times256$ uniform grid in $(s,t)$.  For $R=1000$, the first nonzero shell in \eqref{eq:numerical-fourier} already carries the suppression factor
\[
        e^{-R^2/2}=e^{-500000},
\]
because the smallest positive value of $m^2+n^2-mn$ is $1$.  Thus, to numerical precision, the Gaussian-weighted integrals are indistinguishable from their zero Fourier modes.  Table \ref{tab:d2-numerics} records the resulting values of $\Ical[F_{1000,\eps}]$, $\Qcal[F_{1000,\eps}]$, $\Dcal[F_{1000,\eps}]$, together with the defect
\[
        \Ical[F_{1000,\eps}]\Dcal[F_{1000,\eps}]-\Qcal[F_{1000,\eps}]^2
\]
and the ratio
\[
        \frac{\Ical[F_{1000,\eps}]\Dcal[F_{1000,\eps}]}{\Qcal[F_{1000,\eps}]^2}.
\]

\begin{table}[H]
\caption{Numerical evaluation of the Euclidean family $F_{R,\eps}$ in dimension two for $R=1000$.  In every row we have $\Ical[F_{R,\eps}]\Dcal[F_{R,\eps}]-\Qcal[F_{R,\eps}]^2<0$ and $\Ical[F_{R,\eps}]\Dcal[F_{R,\eps}]/\Qcal[F_{R,\eps}]^2<1$.}
\label{tab:d2-numerics}
\centering
\small
\begin{tabular}{cccccc}
\toprule
$\eps$ & $\Ical[F_{1000,\eps}]$ & $\Qcal[F_{1000,\eps}]$ & $\Dcal[F_{1000,\eps}]$ & $\Ical\Dcal-\Qcal^2$ & $\Ical\Dcal/\Qcal^2$ \\
\midrule
$0.03$ & $1.37209\times 10^{-3}$ & $1.36059\times 10^{-3}$ & $1.34850\times 10^{-3}$ & $-9.47\times 10^{-10}$ & $0.999488$ \\
$0.04$ & $2.44947\times 10^{-3}$ & $2.42547\times 10^{-3}$ & $2.39930\times 10^{-3}$ & $-5.88\times 10^{-9}$ & $0.999001$ \\
$0.05$ & $3.84443\times 10^{-3}$ & $3.80047\times 10^{-3}$ & $3.75429\times 10^{-3}$ & $-1.05\times 10^{-8}$ & $0.999275$ \\
$0.055$ & $4.66233\times 10^{-3}$ & $4.60516\times 10^{-3}$ & $4.54700\times 10^{-3}$ & $-7.90\times 10^{-9}$ & $0.999627$ \\
\bottomrule
\end{tabular}
\end{table}

Each row of Table \ref{tab:d2-numerics} therefore gives an explicit two-dimensional counterexample: for the displayed value of $\eps$, the density $F_{1000,\eps}$ satisfies
\[
        \Ical[F_{1000,\eps}]\Dcal[F_{1000,\eps}]-\Qcal[F_{1000,\eps}]^2<0,
\]
equivalently
\[
        \frac{\Ical[F_{1000,\eps}]\Dcal[F_{1000,\eps}]}{\Qcal[F_{1000,\eps}]^2}<1.
\]
These values are consistent with the perturbative calculation in Proposition \ref{prop:torus} and with the transfer argument of Proposition \ref{prop:euclidean}.  The ratios are close to $1$ because the Gaussian envelope contributes the positive shift $R^{-2}I_2$ to the logarithmic Hessian, but the sign remains strictly negative for all parameter choices listed in the table.

\section{Higher-dimensional Counterexamples}\label{sec:high-d}

There are two natural ways to extend the two-dimensional counterexample to higher dimensions.  The
first is the soft tensorization argument obtained by adjoining broad Gaussian factors.  The second is a
frequency-lattice generalization of the hexagonal zero-mode mechanism on a $d$-dimensional torus.  We record both mechanisms in this
section.  The tensorization argument proves the qualitative higher-dimensional failure.  The direct
construction shows that the same closed-frequency mechanism exists intrinsically in every fixed
dimension.

\subsection{Tensorization of the two-dimensional example}

We first justify Corollary~\ref{cor:high-d} by adjoining a broad Gaussian block in the remaining coordinates.  In the
large-variance limit, this extra factor becomes asymptotically invisible in the quotient
$\mathcal I\mathcal D/\mathcal Q^2$.

\begin{proof}[Proof of Corollary~\ref{cor:high-d}]
Let $F$ be the two-dimensional counterexample from Theorem~\ref{thm:main}.  For $d>2$, let $G_\sigma$ be the
centered Gaussian density on $\mathbb R^{d-2}$ with covariance $\sigma^2 I_{d-2}$, and set
\[
  H_\sigma(x,y)=F(x)G_\sigma(y),
  \qquad x\in\mathbb R^2,\quad y\in\mathbb R^{d-2}.
\]
For product densities, the logarithmic Hessian matrix is block diagonal.  Therefore the three
functionals add:
\[
  \mathcal I[H_\sigma]=\mathcal I[F]+\mathcal I[G_\sigma],\qquad
  \mathcal Q[H_\sigma]=\mathcal Q[F]+\mathcal Q[G_\sigma],\qquad
  \mathcal D[H_\sigma]=\mathcal D[F]+\mathcal D[G_\sigma].
\]
For the Gaussian block,
\[
  \mathcal I[G_\sigma]=(d-2)\sigma^{-2},\qquad
  \mathcal Q[G_\sigma]=(d-2)\sigma^{-4},\qquad
  \mathcal D[G_\sigma]=2(d-2)\sigma^{-6}.
\]
As $\sigma\to\infty$,
\[
  \frac{\mathcal I[H_\sigma]\mathcal D[H_\sigma]}{\mathcal Q[H_\sigma]^2}
  \longrightarrow
  \frac{\mathcal I[F]\mathcal D[F]}{\mathcal Q[F]^2}<1 .
\]
Taking $\sigma$ sufficiently large gives
\[
  \mathcal I[H_\sigma]\mathcal D[H_\sigma]-\mathcal Q[H_\sigma]^2<0.
\]
Define $H_{\sigma,t}=e^{t\Delta/2}H_\sigma$.  As in the proof of Theorem~\ref{thm:main}, the defect
\[
  \Phi_{H_\sigma}(t)
  :=\mathcal I[H_{\sigma,t}]\mathcal D[H_{\sigma,t}]
    -\mathcal Q[H_{\sigma,t}]^2
\]
is continuous at $t=0$, with $\Phi_{H_\sigma}(0)<0$.  Therefore $\Phi_{H_\sigma}(t)<0$ for all
sufficiently small positive $t$, and Lemma~\ref{lem:identities} shows that the corresponding heat flow violates
\eqref{eq:log-convexity}.
\end{proof}

The tensorization argument is enough to prove failure in every dimension $d\ge2$.  However, it hides
the genuinely $d$-dimensional frequency geometry.  We now record a direct construction in
$\mathbb R^d$.  It is equivalent to the usual regular-simplex model, but it avoids passing through an
ambient $(d+1)$-dimensional hyperplane.

\subsection{A higher-dimensional generalization of the hexagonal frequency construction}
\label{subsec:direct-lattice-family}

This second construction keeps the closed-frequency mechanism visible in dimension $d$.  The point of the
construction is that the torus is not chosen first.  We first choose a frequency lattice
$K_d\subset\mathbb R^d$, then take its dual period lattice and work on the corresponding flat torus.

\begin{definition}[Direct frequency lattice]
\label{def:direct-frequency-lattice}
Fix $d\ge2$.  Let $b_1,\ldots,b_d\in\mathbb R^d$ be vectors satisfying
\begin{equation}\label{eq:direct-b-gram}
  b_i\cdot b_j=
  \begin{cases}
    1, & i=j,\\[0.2em]
    -\dfrac12, & |i-j|=1,\\[0.2em]
    0, & |i-j|\ge2.
  \end{cases}
\end{equation}
Define the frequency lattice
\begin{equation}\label{eq:direct-Kd}
  K_d:=\mathbb Z b_1\oplus\cdots\oplus\mathbb Z b_d\subset\mathbb R^d,
\end{equation}
and the dual period lattice
\begin{equation}\label{eq:direct-Lambdad}
  \Lambda_d:=K_d^\vee
  :=\{\lambda\in\mathbb R^d:\xi\cdot\lambda\in2\pi\mathbb Z
        \text{ for all }\xi\in K_d\}.
\end{equation}
We write
\begin{equation}\label{eq:direct-torus}
  \mathbb T_d:=\mathbb R^d/\Lambda_d
\end{equation}
for the associated flat torus, and $\langle\cdot\rangle_d$ for normalized Haar average on
$\mathbb T_d$.
\end{definition}

The definition is intentionally short: it only defines the lattice and the torus.  We next introduce
the active frequencies.  For
\[
  1\le i<j\le d+1,
\]
set
\begin{equation}\label{eq:direct-xiij}
  \xi_{ij}:=\sum_{r=i}^{j-1}b_r.
\end{equation}
The sum in \eqref{eq:direct-xiij} is over the indices $r=i,i+1,\ldots,j-1$.  In particular,
\begin{equation}\label{eq:direct-adjacent-xi}
  \xi_{i,i+1}=b_i,
  \qquad
  \xi_{12}=b_1.
\end{equation}
Thus the vectors $b_1,\ldots,b_d$ are the basis vectors of the frequency lattice, whereas the active
frequencies are all interval sums $\xi_{ij}$.  Their number is
\begin{equation}\label{eq:direct-Nd}
  N_d:=\#\{(i,j):1\le i<j\le d+1\}=\binom{d+1}{2}.
\end{equation}
We shall also use
\begin{equation}\label{eq:direct-Md}
  M_d:=\#\{(i,j,k):1\le i<j<k\le d+1\}=\binom{d+1}{3}.
\end{equation}
Here $M_d$ counts the closed frequency triangles below.

\paragraph{Existence of the Gram model.}
The Gram matrix \eqref{eq:direct-b-gram} is positive definite, and it may be realized explicitly.  If
$e_1,\ldots,e_d$ is the standard orthonormal basis of $\mathbb R^d$, take
\begin{equation}\label{eq:direct-explicit-b}
  b_1=e_1,
  \qquad
  b_i=-\sqrt{\frac{i-1}{2i}}\,e_{i-1}
      +\sqrt{\frac{i+1}{2i}}\,e_i,
  \qquad 2\le i\le d.
\end{equation}
Indeed, $|b_i|^2=(i-1)/(2i)+(i+1)/(2i)=1$.  Moreover,
\[
  b_{i-1}\cdot b_i
  =\sqrt{\frac{i}{2(i-1)}}\left(-\sqrt{\frac{i-1}{2i}}\right)
  =-\frac12,
\]
and $b_i\cdot b_j=0$ when $|i-j|\ge2$.  Hence \eqref{eq:direct-b-gram} holds.

\begin{lemma}[Elementary geometry of the active frequencies]\label{lem:active-frequency-geometry}
The vectors $\xi_{ij}$ have the following properties.
\begin{enumerate}
\item Each active frequency has unit length:
\begin{equation}\label{eq:direct-xi-unit}
  |\xi_{ij}|=1.
\end{equation}
\item For every $1\le i<j<k\le d+1$,
\begin{equation}\label{eq:direct-triangle-relation}
  \xi_{ij}+\xi_{jk}-\xi_{ik}=0.
\end{equation}
\item Every nonzero vector in $K_d$ has length at least one.  More precisely, if
$\xi=\sum_{r=1}^d m_rb_r$ with $m_r\in\mathbb Z$, then
\begin{equation}\label{eq:frequency-gap-direct}
  |\xi|^2
  =\sum_{r=1}^d m_r^2-\sum_{r=1}^{d-1}m_rm_{r+1}
  =\frac12\sum_{r=0}^{d}(m_{r+1}-m_r)^2,
\end{equation}
where $m_0=m_{d+1}=0$.  Consequently $|\xi|^2\ge1$ for every $0\ne\xi\in K_d$, and equality is attained
precisely by the shortest vectors $\pm\xi_{ij}$.
\end{enumerate}
\end{lemma}

\begin{proof}
For \eqref{eq:direct-xi-unit}, fix $1\le i<j\le d+1$ and put $m=j-i$.  Then
\[
  \xi_{ij}=b_i+b_{i+1}+\cdots+b_{j-1}
\]
contains $m$ consecutive vectors.  Therefore
\begin{align*}
  |\xi_{ij}|^2
  &=\sum_{r=i}^{j-1}|b_r|^2
    +2\sum_{i\le r<s\le j-1}b_r\cdot b_s.
\end{align*}
The first sum is $m$.  In the second sum, the only nonzero terms are neighboring pairs
$b_r\cdot b_{r+1}=-1/2$, and there are $m-1$ such pairs.  Hence
\[
  |\xi_{ij}|^2=m+2(m-1)\left(-\frac12\right)=1.
\]
This proves \eqref{eq:direct-xi-unit}.

For $i<j<k$, the definition gives
\[
  \xi_{ij}=\sum_{r=i}^{j-1}b_r,
  \qquad
  \xi_{jk}=\sum_{r=j}^{k-1}b_r,
  \qquad
  \xi_{ik}=\sum_{r=i}^{k-1}b_r.
\]
Therefore $\xi_{ij}+\xi_{jk}=\xi_{ik}$, proving \eqref{eq:direct-triangle-relation}.

Finally, for $\xi=\sum_{r=1}^d m_rb_r$, the first equality in
\eqref{eq:frequency-gap-direct} follows by expanding with the Gram matrix \eqref{eq:direct-b-gram}.
For the second equality, set $m_0=m_{d+1}=0$ and expand
\[
  \sum_{r=0}^{d}(m_{r+1}-m_r)^2
  =2\sum_{r=1}^d m_r^2-2\sum_{r=1}^{d-1}m_rm_{r+1}.
\]
Since the $m_r$ are integers, a nonzero sequence
$0=m_0,m_1,\ldots,m_d,m_{d+1}=0$ must have at least two nonzero jumps.  Thus
$\sum_{r=0}^d(m_{r+1}-m_r)^2\ge2$, and so $|\xi|^2\ge1$.  Equality holds exactly when there are two
jumps of size $\pm1$ and all other jumps are zero.  Equivalently, $(m_r)$ is equal to $1$ or $-1$ on
one nonempty interval and zero outside.  This gives precisely $\xi=\pm\sum_{r=i}^{j-1}b_r=\pm\xi_{ij}$.
\end{proof}

Now define the periodic potential
\begin{equation}\label{eq:direct-psid}
  \psi_d(x):=\sum_{1\le i<j\le d+1}\cos(\xi_{ij}\cdot x),
  \qquad x\in\mathbb T_d.
\end{equation}
Let
\begin{equation}\label{eq:direct-Rplus}
  \mathcal R_d^+:=\{\xi_{ij}:1\le i<j\le d+1\}.
\end{equation}
For $\alpha\in\mathcal R_d^+$, write
\begin{equation}\label{eq:direct-cs-alpha}
  c_\alpha(x):=\cos(\alpha\cdot x),
  \qquad
  s_\alpha(x):=\sin(\alpha\cdot x).
\end{equation}
Thus
\[
  \psi_d=\sum_{\alpha\in\mathcal R_d^+}c_\alpha.
\]
Since $|\alpha|=1$ for every active frequency,
\begin{equation}\label{eq:direct-minus-laplacian-psi}
  -\Delta\psi_d=\psi_d.
\end{equation}
For $\varepsilon$ small, set
\begin{equation}\label{eq:direct-fdepsilon}
  f_{d,\varepsilon}=Z_{d,\varepsilon}^{-1}e^{\varepsilon\psi_d},
  \qquad
  Z_{d,\varepsilon}:=\langle e^{\varepsilon\psi_d}\rangle_d.
\end{equation}
Then
\begin{equation}\label{eq:direct-H-fdepsilon}
  \mathsf H_{f_{d,\varepsilon}}
  =-\nabla^2\log f_{d,\varepsilon}
  =-\varepsilon\nabla^2\psi_d.
\end{equation}

The following proposition is the detailed expansion underlying this higher-dimensional hexagonal frequency construction.

\begin{proposition}[Higher-dimensional hexagonal frequency-lattice expansion]\label{prop:direct-lattice-expansion}
On the torus $\mathbb T_d$, one has
\begin{align}
  \mathcal I[f_{d,\varepsilon}]
  &=\frac{N_d}{2}\varepsilon^2+\frac{3M_d}{4}\varepsilon^3+O_d(\varepsilon^4),
  \label{eq:direct-I-exp}\\
  \mathcal Q[f_{d,\varepsilon}]
  &=\frac{N_d}{2}\varepsilon^2+\frac{3M_d}{8}\varepsilon^3+O_d(\varepsilon^4),
  \label{eq:direct-Q-exp}\\
  \mathcal D[f_{d,\varepsilon}]
  &=\frac{N_d}{2}\varepsilon^2-\frac{3M_d}{16}\varepsilon^3+O_d(\varepsilon^4).
  \label{eq:direct-D-exp}
\end{align}
Consequently,
\begin{equation}\label{eq:direct-defect-expansion}
  \mathcal I[f_{d,\varepsilon}]\mathcal D[f_{d,\varepsilon}]
  -\mathcal Q[f_{d,\varepsilon}]^2
  =-\frac{3N_dM_d}{32}\varepsilon^5+O_d(\varepsilon^6),
\end{equation}
and
\begin{equation}\label{eq:direct-ratio-expansion}
  \frac{\mathcal I[f_{d,\varepsilon}]\mathcal D[f_{d,\varepsilon}]}
       {\mathcal Q[f_{d,\varepsilon}]^2}
  =1-\frac{d-1}{8}\varepsilon+O_d(\varepsilon^2).
\end{equation}
\end{proposition}

\begin{proof}
We give the details because the constants in the expansion come from two separate mechanisms: the
zero-mode selection on the torus and the metric contractions of the frequency vectors.

\paragraph{Step 1: Fourier selection and cubic zero modes.}
By the definition of $\Lambda_d$, the functions $e^{i\xi\cdot x}$, $\xi\in K_d$, are precisely the
Fourier characters on $\mathbb T_d$.  Hence
\begin{equation}\label{eq:direct-fourier-selection}
  \langle e^{i\xi\cdot x}\rangle_d=
  \begin{cases}
    1,&\xi=0,\\
    0,&\xi\ne0.
  \end{cases}
\end{equation}
For $\alpha,\beta,\gamma\in\mathcal R_d^+$,
\[
  c_\alpha c_\beta c_\gamma
  =\frac18\sum_{\sigma_1,\sigma_2,\sigma_3\in\{\pm1\}}
  e^{i(\sigma_1\alpha+\sigma_2\beta+\sigma_3\gamma)\cdot x}.
\]
Therefore a cubic average can be nonzero only if a signed sum of the three active frequencies is zero.
For the interval frequencies \eqref{eq:direct-xiij}, such a relation is necessarily one of the closed
triangles
\begin{equation}\label{eq:direct-cubic-triangle}
  \xi_{ij}+\xi_{jk}-\xi_{ik}=0,
  \qquad 1\le i<j<k\le d+1.
\end{equation}
Indeed, in the $b$-basis the vector $\xi_{ij}$ is the indicator of the interval $[i,j-1]$.  A relation
of three signed interval indicators can be rearranged as
\[
  \mathbf 1_A+\mathbf 1_B=\mathbf 1_C.
\]
The two intervals $A,B$ must be disjoint, and their union must again be an interval; hence they are
adjacent intervals.  This is exactly \eqref{eq:direct-cubic-triangle}.

\paragraph{Step 2: the basic averages of $\psi_d$.}
For every active frequency,
\[
  \langle c_\alpha\rangle_d=0,
  \qquad
  \langle c_\alpha^2\rangle_d=\frac12,
\]
and for distinct active frequencies $\alpha\ne\beta$,
\[
  \langle c_\alpha c_\beta\rangle_d=0.
\]
Therefore
\begin{equation}\label{eq:direct-psi-second}
  \langle\psi_d\rangle_d=0,
  \qquad
  \langle\psi_d^2\rangle_d=\frac{N_d}{2}.
\end{equation}
Next fix one closed triangle
\[
  p=\xi_{ij},\qquad q=\xi_{jk},\qquad r=\xi_{ik}=p+q.
\]
If $A=p\cdot x$ and $B=q\cdot x$, then $r\cdot x=A+B$.  Hence
\begin{align*}
  \cos A\cos B\cos(A+B)
  &=\frac12\cos^2(A+B)+\frac12\cos(A-B)\cos(A+B)\\
  &=\frac14\{1+\cos(2A)+\cos(2B)+\cos(2A+2B)\}.
\end{align*}
Averaging kills the last three terms, so
\begin{equation}\label{eq:direct-triangle-average}
  \langle c_pc_qc_r\rangle_d=\frac14.
\end{equation}
In $\psi_d^3$, each unordered triangle gives six ordered products.  Since there are $M_d$ closed
triangles,
\begin{equation}\label{eq:direct-psi-third}
  \langle\psi_d^3\rangle_d
  =M_d\cdot6\cdot\frac14
  =\frac{3M_d}{2}.
\end{equation}

\paragraph{Step 3: Hessian identities and contractions.}
For $\alpha\in\mathcal R_d^+$,
\[
  \nabla^2\cos(\alpha\cdot x)=-c_\alpha\,\alpha\otimes\alpha.
\]
It is convenient to set
\begin{equation}\label{eq:direct-K-definition}
  \mathsf K_d:=-\nabla^2\psi_d
  =\sum_{\alpha\in\mathcal R_d^+}c_\alpha\,\alpha\otimes\alpha.
\end{equation}
Then
\begin{equation}\label{eq:direct-grad-K}
  \nabla\mathsf K_d
  =-\sum_{\alpha\in\mathcal R_d^+}s_\alpha\,\alpha^{\otimes3}.
\end{equation}
We shall use
\begin{equation}\label{eq:direct-rank-one-contractions}
  (\alpha\otimes\alpha):(\beta\otimes\beta)=(\alpha\cdot\beta)^2,
\end{equation}
\begin{equation}\label{eq:direct-third-contraction}
  \alpha^{\otimes3}:\beta^{\otimes3}=(\alpha\cdot\beta)^3,
\end{equation}
and
\begin{equation}\label{eq:direct-trace-contraction}
  \operatorname{tr}\{(\alpha\otimes\alpha)(\beta\otimes\beta)(\gamma\otimes\gamma)\}
  = (\alpha\cdot\beta)(\beta\cdot\gamma)(\gamma\cdot\alpha).
\end{equation}
For a closed triangle $p+q-r=0$ with $|p|=|q|=|r|=1$, we have $r=p+q$, hence
\begin{equation}\label{eq:direct-triangle-inner-products}
  p\cdot q=-\frac12,
  \qquad
  p\cdot r=q\cdot r=\frac12.
\end{equation}

\paragraph{Step 4: expansion of $\mathcal I$.}
Since $-\Delta\psi_d=\psi_d$ and $\mathsf H_{f_{d,\varepsilon}}=\varepsilon\mathsf K_d$, we have
\[
  \operatorname{tr}(\mathsf H_{f_{d,\varepsilon}})=\varepsilon\psi_d.
\]
Thus
\[
  \mathcal I[f_{d,\varepsilon}]
  =\varepsilon\frac{\langle\psi_d e^{\varepsilon\psi_d}\rangle_d}
                       {\langle e^{\varepsilon\psi_d}\rangle_d}.
\]
Using
\[
  e^{\varepsilon\psi_d}=1+\varepsilon\psi_d+\frac{\varepsilon^2}{2}\psi_d^2+O_d(\varepsilon^3),
\]
we get
\begin{align*}
  \langle\psi_d e^{\varepsilon\psi_d}\rangle_d
  &=\langle\psi_d\rangle_d
    +\varepsilon\langle\psi_d^2\rangle_d
    +\frac{\varepsilon^2}{2}\langle\psi_d^3\rangle_d
    +O_d(\varepsilon^3)\\
  &=\frac{N_d}{2}\varepsilon+\frac{3M_d}{4}\varepsilon^2+O_d(\varepsilon^3).
\end{align*}
Moreover,
\[
  \langle e^{\varepsilon\psi_d}\rangle_d=1+O_d(\varepsilon^2),
\]
so the denominator does not affect the coefficients through order $\varepsilon^3$.  This proves
\eqref{eq:direct-I-exp}.

\paragraph{Step 5: expansion of $\mathcal Q$.}
Since
\[
  \operatorname{tr}(\mathsf H_{f_{d,\varepsilon}}^2)=\varepsilon^2|\mathsf K_d|^2,
\]
we have
\begin{align*}
  \mathcal Q[f_{d,\varepsilon}]
  =\varepsilon^2\frac{\langle |\mathsf K_d|^2 e^{\varepsilon\psi_d}\rangle_d}
                         {\langle e^{\varepsilon\psi_d}\rangle_d}
  =\varepsilon^2\langle |\mathsf K_d|^2\rangle_d
   +\varepsilon^3\langle\psi_d|\mathsf K_d|^2\rangle_d
   +O_d(\varepsilon^4).
\end{align*}
From \eqref{eq:direct-K-definition} and \eqref{eq:direct-rank-one-contractions},
\[
  |\mathsf K_d|^2
  =\sum_{\alpha,\beta\in\mathcal R_d^+}
    c_\alpha c_\beta(\alpha\cdot\beta)^2.
\]
Only the diagonal terms survive in the average, so
\begin{equation}\label{eq:direct-K-square-average}
  \langle |\mathsf K_d|^2\rangle_d
  =\sum_{\alpha\in\mathcal R_d^+}\frac12|\alpha|^4
  =\frac{N_d}{2}.
\end{equation}
For the cubic term, expand
\[
  \psi_d|\mathsf K_d|^2
  =\sum_{\gamma,\alpha,\beta}
    c_\gamma c_\alpha c_\beta(\alpha\cdot\beta)^2.
\]
Only closed triangles contribute.  For one triangle $\{p,q,r\}$, there are six ordered triples.  In
each of them, the trigonometric average is $1/4$, and the squared inner product is $(1/2)^2=1/4$.
Therefore each triangle contributes
\[
  6\cdot\frac14\cdot\frac14=\frac38.
\]
Thus
\begin{equation}\label{eq:direct-psi-K-square-average}
  \langle\psi_d|\mathsf K_d|^2\rangle_d=\frac{3M_d}{8}.
\end{equation}
Combining \eqref{eq:direct-K-square-average} and \eqref{eq:direct-psi-K-square-average} proves
\eqref{eq:direct-Q-exp}.

\paragraph{Step 6: expansion of $\mathcal D$.}
Using $\mathsf H_{f_{d,\varepsilon}}=\varepsilon\mathsf K_d$,
\begin{align*}
  \mathcal D[f_{d,\varepsilon}]
  &=\frac{\langle (\varepsilon^2|\nabla\mathsf K_d|^2
       +2\varepsilon^3\operatorname{tr}(\mathsf K_d^3))e^{\varepsilon\psi_d}\rangle_d}
       {\langle e^{\varepsilon\psi_d}\rangle_d}\\
  &=\varepsilon^2\langle |\nabla\mathsf K_d|^2\rangle_d
    +\varepsilon^3\left(
      \langle\psi_d|\nabla\mathsf K_d|^2\rangle_d
      +2\langle\operatorname{tr}(\mathsf K_d^3)\rangle_d
      \right)
    +O_d(\varepsilon^4).
\end{align*}
From \eqref{eq:direct-grad-K} and \eqref{eq:direct-third-contraction},
\[
  |\nabla\mathsf K_d|^2
  =\sum_{\alpha,\beta}s_\alpha s_\beta(\alpha\cdot\beta)^3.
\]
Only the diagonal terms survive in the average, and hence
\begin{equation}\label{eq:direct-grad-K-average}
  \langle |\nabla\mathsf K_d|^2\rangle_d
  =\sum_{\alpha}\frac12|\alpha|^6=\frac{N_d}{2}.
\end{equation}
We now compute the cubic part.  Fix a closed triangle $p+q-r=0$ and set
$A=p\cdot x$, $B=q\cdot x$, so $r\cdot x=A+B$.  The three elementary averages are
\begin{equation}\label{eq:direct-sine-cosine-averages}
  \langle c_rs_ps_q\rangle_d=-\frac14,
  \qquad
  \langle c_qs_ps_r\rangle_d=\frac14,
  \qquad
  \langle c_ps_qs_r\rangle_d=\frac14.
\end{equation}
For instance,
\[
  \cos(A+B)\sin A\sin B
  =\sin A\cos A\sin B\cos B-\sin^2 A\sin^2 B,
\]
whose average is $-1/4$.  The other two identities are similar.  Combining
\eqref{eq:direct-sine-cosine-averages} with \eqref{eq:direct-triangle-inner-products}, the contribution
of one triangle to $\langle\psi_d|\nabla\mathsf K_d|^2\rangle_d$ is
\[
  2\left(-\frac14\right)\left(-\frac18\right)
  +2\left(\frac14\right)\left(\frac18\right)
  +2\left(\frac14\right)\left(\frac18\right)
  =\frac3{16}.
\]
Therefore
\begin{equation}\label{eq:direct-psi-grad-K-average}
  \langle\psi_d|\nabla\mathsf K_d|^2\rangle_d=\frac{3M_d}{16}.
\end{equation}
Next, by \eqref{eq:direct-trace-contraction},
\[
  \operatorname{tr}(\mathsf K_d^3)
  =\sum_{\alpha,\beta,\gamma}
  c_\alpha c_\beta c_\gamma
  (\alpha\cdot\beta)(\beta\cdot\gamma)(\gamma\cdot\alpha).
\]
For one closed triangle, each of the six ordered triples has trigonometric average $1/4$ and inner
product factor
\[
  (p\cdot q)(q\cdot r)(r\cdot p)
  =\left(-\frac12\right)\left(\frac12\right)\left(\frac12\right)
  =-\frac18.
\]
Thus each triangle contributes
\[
  6\cdot\frac14\left(-\frac18\right)=-\frac3{16},
\]
and hence
\begin{equation}\label{eq:direct-tr-K-cube-average}
  \langle\operatorname{tr}(\mathsf K_d^3)\rangle_d=-\frac{3M_d}{16}.
\end{equation}
Using \eqref{eq:direct-grad-K-average}, \eqref{eq:direct-psi-grad-K-average}, and
\eqref{eq:direct-tr-K-cube-average}, the cubic coefficient in $\mathcal D$ is
\[
  \frac{3M_d}{16}+2\left(-\frac{3M_d}{16}\right)=-\frac{3M_d}{16}.
\]
This proves \eqref{eq:direct-D-exp}.

\paragraph{Step 7: the defect and the quotient.}
Write
\[
  \mathcal I=a_2\varepsilon^2+a_3\varepsilon^3+O_d(\varepsilon^4),
  \quad
  \mathcal Q=q_2\varepsilon^2+q_3\varepsilon^3+O_d(\varepsilon^4),
  \quad
  \mathcal D=d_2\varepsilon^2+d_3\varepsilon^3+O_d(\varepsilon^4),
\]
where
\[
  a_2=q_2=d_2=\frac{N_d}{2},
  \quad
  a_3=\frac{3M_d}{4},
  \quad
  q_3=\frac{3M_d}{8},
  \quad
  d_3=-\frac{3M_d}{16}.
\]
Then
\begin{align*}
  \mathcal I\mathcal D-\mathcal Q^2
  &=(a_2d_2-q_2^2)\varepsilon^4
    +(a_2d_3+a_3d_2-2q_2q_3)\varepsilon^5
    +O_d(\varepsilon^6).
\end{align*}
The $\varepsilon^4$ coefficient is zero.  The $\varepsilon^5$ coefficient is
\[
  \frac{N_d}{2}\left(-\frac{3M_d}{16}\right)
  +\frac{3M_d}{4}\frac{N_d}{2}
  -2\frac{N_d}{2}\frac{3M_d}{8}
  =-\frac{3N_dM_d}{32}.
\]
This proves \eqref{eq:direct-defect-expansion}.  Finally,
\[
  \mathcal Q^2=\frac{N_d^2}{4}\varepsilon^4+O_d(\varepsilon^5),
\]
so
\[
  \frac{\mathcal I\mathcal D}{\mathcal Q^2}
  =1-\frac{3M_d}{8N_d}\varepsilon+O_d(\varepsilon^2).
\]
Since
\[
  \frac{M_d}{N_d}
  =\frac{\binom{d+1}{3}}{\binom{d+1}{2}}
  =\frac{d-1}{3},
\]
we obtain \eqref{eq:direct-ratio-expansion}.
\end{proof}

\paragraph{Counting clarification.}
The rank of the frequency lattice is $d$, and $b_1,\ldots,b_d$ are its basis vectors.  However, the
perturbation does not use only these $d$ basis frequencies.  It uses all interval sums
$\xi_{ij}=\sum_{r=i}^{j-1}b_r$, i.e. all shortest nonzero frequency vectors.  Therefore the number of
active cosine modes is
\[
  N_d=\binom{d+1}{2},
\]
not $d$.  The cubic zero modes are indexed by triples $i<j<k$, so their number is
\[
  M_d=\binom{d+1}{3}.
\]
For $d=2$ this gives $N_2=3$ and $M_2=1$, recovering exactly the hexagonal perturbation from the main
two-dimensional construction.

\begin{corollary}[Euclidean realizations of the generalized hexagonal family]\label{cor:direct-lattice-euclidean}
For every fixed $d\ge2$, there exists a smooth, strictly positive, Gaussian-decaying probability density
$F$ on $\mathbb R^d$ arising from the generalized hexagonal frequency-lattice family such that
\[
  \mathcal I[F]\mathcal D[F]-\mathcal Q[F]^2<0.
\]
\end{corollary}

\begin{proof}
For $R>0$, define
\[
  F_{R,\varepsilon,d}(x)
  =Z_{R,\varepsilon,d}^{-1}
   \exp\left(\varepsilon\psi_d(x)-\frac{|x|^2}{2R^2}\right),
  \qquad x\in\mathbb R^d.
\]
With $\mathsf K_d:=-\nabla^2\psi_d$ as in \eqref{eq:direct-K-definition}, its logarithmic Hessian is
\[
  \mathsf H_{F_{R,\varepsilon,d}}
  =\varepsilon\mathsf K_d+R^{-2}I_d.
\]
The constant shift $R^{-2}I_d$ disappears as $R\to\infty$.  Moreover, by the frequency gap
\eqref{eq:frequency-gap-direct}, every nonzero Fourier mode of the lattice is suppressed by a factor
at most $e^{-R^2/2}$ under the Gaussian envelope.  Thus the same averaging argument as in
Lemma~\ref{lem:Gaussian-envelope} gives
\[
  \mathcal I[F_{R,\varepsilon,d}]\to \mathcal I[f_{d,\varepsilon}],
  \qquad
  \mathcal Q[F_{R,\varepsilon,d}]\to \mathcal Q[f_{d,\varepsilon}],
  \qquad
  \mathcal D[F_{R,\varepsilon,d}]\to \mathcal D[f_{d,\varepsilon}]
\]
as $R\to\infty$.  For sufficiently small positive $\varepsilon$,
Proposition~\ref{prop:direct-lattice-expansion} gives a strictly negative torus defect.  The strict
inequality therefore persists for $F_{R,\varepsilon,d}$ when $R$ is sufficiently large.
\end{proof}

\section{Sharp Constants and Open Problems}
\label{sec:open-problems}
Let $\mathcal S(\R^d)$ denote the Schwartz class, namely the smooth functions whose derivatives decay faster than every polynomial rate.  Define the admissible class in dimension $d$ by
\[
        \mathcal A_d
        :=
        \left\{
        f\in\mathcal S(\R^d):
        f>0,\quad
        \int_{\R^d}f\,\dd x=1,\quad
        \Ical[f],\Qcal[f],\Dcal[f]<\infty,\quad
        \Qcal[f]>0
        \right\}.
\]
A member of $\mathcal A_d$ will be called an admissible density.  The functionals are computed with $\Hmat_f=-\nabla^2\log f$ as above.  For $d\ge1$, define the best constant
\begin{equation}\label{eq:theta-d}
        \theta_d^*=
        \inf_{f\in\mathcal A_d}\frac{\Ical[f]\Dcal[f]}{\Qcal[f]^2},
\end{equation}
This quantity is the sharp constant in the family of inequalities
\[
        \Ical[f]\Dcal[f]\ge \theta\,\Qcal[f]^2,
\]
and therefore provides a natural quantitative replacement for the false log-convexity conjecture in dimensions $d\ge2$.

\subsection{Low-dimensional sharp constants and monotonicity}

The one-dimensional log-convexity theorem of Ledoux, Nair, and Wang yields the lower bound $\theta_1^*\ge1$.  The constant $1$ is in fact sharp.

\begin{theorem}[Sharp constant in dimension one]\label{prop:theta-one}
One has
\[
        \theta_1^*=1.
\]
\end{theorem}

\begin{proof}
The one-dimensional log-convexity theorem of \cite{LedouxNairWang} gives
\[
        \Ical[f]\Dcal[f]\ge\Qcal[f]^2
\]
for every admissible density $f$ on $\R$.  Hence $\theta_1^*\ge1$.

To prove the reverse inequality, begin on the flat torus $\T=\R/2\pi\Z$.  For $\eps\in\R$, set
\[
        Z_\eps=\av{e^{\eps\cos x}},
        \qquad
        f_\eps^{\T}(x)=Z_\eps^{-1}e^{\eps\cos x},
        \qquad x\in\T,
\]
where $\av{\cdot}$ denotes normalized Haar average on $\T$.  Since
\[
        \log f_\eps^{\T}(x)=\eps\cos x-\log Z_\eps,
\]
the one-dimensional logarithmic Hessian is
\[
        \Hmat_{f_\eps^{\T}}(x)=\eps\cos x.
\]
Therefore
\[
        \Ical_{\T}(\eps)
        :=
        \Ical[f_\eps^{\T}]
        =
        \frac{\av{\eps\cos x\,e^{\eps\cos x}}}{Z_\eps},
\]
\[
        \Qcal_{\T}(\eps)
        :=
        \Qcal[f_\eps^{\T}]
        =
        \frac{\av{\eps^2\cos^2x\,e^{\eps\cos x}}}{Z_\eps},
\]
and
\[
        \Dcal_{\T}(\eps)
        :=
        \Dcal[f_\eps^{\T}]
        =
        \frac{\av{\bigl(\eps^2\sin^2x+2\eps^3\cos^3x\bigr)e^{\eps\cos x}}}{Z_\eps}.
\]
Using
\[
        \av{\cos x}=0,
        \qquad
        \av{\cos^2x}=\av{\sin^2x}=\frac12,
        \qquad
        \av{\cos^3x}=0,
\]
and Taylor expansion of $e^{\eps\cos x}$ at $\eps=0$, we obtain
\[
        Z_\eps=1+O(\eps^2),
\]
\[
        \av{\cos x\,e^{\eps\cos x}}
        =
        \eps\,\av{\cos^2x}+O(\eps^3)
        =
        \frac{\eps}{2}+O(\eps^3),
\]
\[
        \av{\cos^2x\,e^{\eps\cos x}}
        =
        \av{\cos^2x}+O(\eps^2)
        =
        \frac12+O(\eps^2),
\]
and
\[
        \av{\bigl(\sin^2x+2\eps\cos^3x\bigr)e^{\eps\cos x}}
        =
        \av{\sin^2x}+O(\eps^2)
        =
        \frac12+O(\eps^2).
\]
Hence
\[
        \Ical_{\T}(\eps)=\frac12\eps^2+O(\eps^4),
        \qquad
        \Qcal_{\T}(\eps)=\frac12\eps^2+O(\eps^4),
        \qquad
        \Dcal_{\T}(\eps)=\frac12\eps^2+O(\eps^4),
\]
and therefore
\begin{equation}\label{eq:theta-one-torus}
        \frac{\Ical_{\T}(\eps)\Dcal_{\T}(\eps)}{\Qcal_{\T}(\eps)^2}
        =
        1+O(\eps^2)
        \qquad (\eps\downarrow0).
\end{equation}

We next transfer this family to $\R$.  For $R>0$, define
\[
        F_{R,\eps}(x)=Z_{R,\eps}^{-1}\exp\left(\eps\cos x-\frac{x^2}{2R^2}\right),
        \qquad x\in\R,
\]
where
\[
        Z_{R,\eps}=\int_\R e^{\eps\cos x}e^{-x^2/(2R^2)}\dd x.
\]
Then $F_{R,\eps}$ is a smooth strictly positive Gaussian-decaying density on $\R$, and
\[
        \Hmat_{F_{R,\eps}}(x)=\eps\cos x+R^{-2},
        \qquad
        \partial_x\Hmat_{F_{R,\eps}}(x)=-\eps\sin x.
\]
If $G$ is any smooth $2\pi$-periodic function with Fourier series
\[
        G(x)=\sum_{m\in\Z}\widehat G(m)e^{imx},
\]
then
\begin{equation}\label{eq:one-d-gaussian-average}
        \int_\R G(x)e^{-x^2/(2R^2)}\dd x
        =
        \sqrt{2\pi}\,R\sum_{m\in\Z}\widehat G(m)e^{-m^2R^2/2}
        =
        \sqrt{2\pi}\,R\,\av{G}+o(R)
        \qquad (R\to\infty).
\end{equation}
Applying \eqref{eq:one-d-gaussian-average} to $e^{\eps\cos x}$ and to the periodic integrands defining $\Ical[F_{R,\eps}]$, $\Qcal[F_{R,\eps}]$, and $\Dcal[F_{R,\eps}]$, and noting that the shift $R^{-2}$ in $\Hmat_{F_{R,\eps}}$ contributes only $o(1)$ after normalization, yields
\[
        \Ical[F_{R,\eps}]\to \Ical_{\T}(\eps),
        \qquad
        \Qcal[F_{R,\eps}]\to \Qcal_{\T}(\eps),
        \qquad
        \Dcal[F_{R,\eps}]\to \Dcal_{\T}(\eps)
        \qquad (R\to\infty)
\]
for every fixed $\eps$.  Hence
\[
        \lim_{R\to\infty}
        \frac{\Ical[F_{R,\eps}]\Dcal[F_{R,\eps}]}{\Qcal[F_{R,\eps}]^2}
        =
        \frac{\Ical_{\T}(\eps)\Dcal_{\T}(\eps)}{\Qcal_{\T}(\eps)^2}.
\]
Combining this with \eqref{eq:theta-one-torus}, we may choose a sequence $\eps_n\downarrow0$ and then $R_n\to\infty$ such that
\[
        \frac{\Ical[F_{R_n,\eps_n}]\Dcal[F_{R_n,\eps_n}]}{\Qcal[F_{R_n,\eps_n}]^2}
        \longrightarrow1.
\]
Since each $F_{R_n,\eps_n}$ is admissible on $\R$, this proves $\theta_1^*\le1$.  Together with $\theta_1^*\ge1$, the conclusion follows.
\end{proof}

For the centered Gaussian density with variance $\sigma^2$, one has $\Hmat_f=\sigma^{-2}$, hence
\[
        \Ical[f]=\sigma^{-2},
        \qquad
        \Qcal[f]=\sigma^{-4},
        \qquad
        \Dcal[f]=2\sigma^{-6},
\]
and therefore
\[
        \frac{\Ical[f]\Dcal[f]}{\Qcal[f]^2}=2.
\]
Thus the sharp constant $\theta_1^*=1$ is not attained by the Gaussian family; it is approached instead by near-flat unit-frequency perturbations.

Corollary \ref{cor:high-d} gives
\begin{equation}\label{eq:theta-less-one}
        \theta_d^*<1,
        \qquad d\ge2.
\end{equation}
In dimension two, combining Corollary \ref{cor:high-d} with the square-root convexity theorem of Liu and Gao \cite{LiuGao2023} yields
\begin{equation}\label{eq:theta-two-range}
        \frac12\le \theta_2^*<1.
\end{equation}
This is the first rigorous localization of a nontrivial sharp constant beyond dimension one.  At the same time, the upper bound in \eqref{eq:theta-two-range} is only qualitative: the hexagonal perturbation constructed in the present paper shows merely that $\theta_2^*$ lies strictly below $1$, and it leaves open the possibility that a nonperturbative family may drive the quotient substantially lower.

We next record a general structural property of the sequence $\{\theta_d^*\}_{d\ge1}$.

\begin{proposition}[Monotonicity in the dimension]\label{prop:theta-monotone}
For every $d\ge1$,
\begin{equation}\label{eq:theta-monotone}
        \theta_{d+1}^*\le \theta_d^*.
\end{equation}
\end{proposition}

\begin{proof}
Fix $d\ge1$ and let $f$ be an admissible density on $\R^d$.  For $\sigma>0$, let
\[
        G_\sigma(y)=(2\pi\sigma^2)^{-1/2}e^{-y^2/(2\sigma^2)},
        \qquad y\in\R,
\]
and define
\[
        F_\sigma(x,y)=f(x)G_\sigma(y),
        \qquad (x,y)\in\R^d\times\R.
\]
Then $F_\sigma$ is an admissible density on $\R^{d+1}$.  Since
\[
        \log F_\sigma(x,y)=\log f(x)-\frac{y^2}{2\sigma^2}-\frac12\log(2\pi\sigma^2),
\]
its logarithmic Hessian is block diagonal:
\[
        \Hmat_{F_\sigma}(x,y)=
        \begin{pmatrix}
                \Hmat_f(x) & 0\\
                0 & \sigma^{-2}
        \end{pmatrix}.
\]
Moreover, $\nabla\Hmat_{F_\sigma}$ has the same $x$-derivatives as $\nabla\Hmat_f$, because the Gaussian block is constant.  Therefore
\[
        \Ical[F_\sigma]=\Ical[f]+\sigma^{-2},
\]
\[
        \Qcal[F_\sigma]=\Qcal[f]+\sigma^{-4},
\]
and
\[
        \Dcal[F_\sigma]=\Dcal[f]+2\sigma^{-6}.
\]
Hence
\[
        \frac{\Ical[F_\sigma]\Dcal[F_\sigma]}{\Qcal[F_\sigma]^2}
        \longrightarrow
        \frac{\Ical[f]\Dcal[f]}{\Qcal[f]^2}
        \qquad (\sigma\to\infty).
\]
Fix $\eta>0$.  Choose $f$ so that
\[
        \frac{\Ical[f]\Dcal[f]}{\Qcal[f]^2}<\theta_d^*+\eta.
\]
For $\sigma$ large enough,
\[
        \frac{\Ical[F_\sigma]\Dcal[F_\sigma]}{\Qcal[F_\sigma]^2}
        <\theta_d^*+2\eta.
\]
Since $F_\sigma$ is admissible in dimension $d+1$, we get
\[
        \theta_{d+1}^*\le\theta_d^*+2\eta.
\]
Letting $\eta\downarrow0$ proves \eqref{eq:theta-monotone}.
\end{proof}

It follows from \eqref{eq:theta-monotone} that the limit
\begin{equation}\label{eq:theta-infty}
        \theta_\infty^*:=\inf_{d\ge2}\theta_d^*=\lim_{d\to\infty}\theta_d^*
\end{equation}
exists in the extended interval $[-\infty,1)$.  Determining the exact values of the constants $\theta_d^*$ is therefore a natural replacement for the false log-convexity conjecture.  The monotonicity argument above gives only the soft upper bound $\theta_\infty^*\le\theta_2^*<1$, and its proof uses only broad-Gaussian tensorization, so it does not capture any intrinsically high-dimensional closed-frequency interaction.

\subsection{A global dichotomy governed by Fisher convexity}

We next isolate the basic dichotomy governing possible lower bounds for the sharp constants.  The decisive quantity is the sign of $\Dcal$ itself, rather than the sharper defect $\Ical\Dcal-\Qcal^2$.

\begin{theorem}[A dichotomy for $\theta_\infty^*$]\label{prop:theta-dichotomy}
The following statements hold.
\begin{enumerate}[label=(\alph*)]
\item If $\Dcal[f]\ge0$ for every admissible density $f$ in every dimension $d\ge2$, then
\[
        0\le \theta_\infty^*<1.
\]
\item If there exist $d_0\ge1$ and an admissible density $g$ on $\R^{d_0}$ such that
\[
        \Dcal[g]<0,
\]
then
\[
        \theta_d^*=-\infty
        \qquad\text{for every }d\ge d_0.
\]
In particular, $\theta_\infty^*=-\infty$.
\end{enumerate}
\end{theorem}

The mechanism behind part (b) is a separated-mixture construction.  One starts from a fixed background density $h$, inserts a very narrow copy of a ``bad'' density $g$ at a distant location, and then tunes three parameters: the small scale $r$, the small mass $\eta$, and the large separation distance $L$.  The role of $r$ is to amplify $\Qcal$ and $\Dcal$ through their scaling laws; the role of $\eta$ is to keep the perturbation negligible at the level of $\Ical$; and the role of $L$ is to decouple the two components so that the functionals become asymptotically additive.

\begin{figure}[H]
\centering
\begin{tikzpicture}[x=1cm,y=1cm]
\begin{scope}
\clip (-0.39,0) rectangle (10.39,2.35);
\draw[thick,blue,smooth,variable=\x,domain=-0.39:10.39,samples=300]
  plot ({\x},{2.15*exp(-0.5*((\x-2.30)/0.90)^2)});
\draw[thick,red,smooth,variable=\x,domain=-0.39:10.39,samples=300]
  plot ({\x},{1.85*exp(-0.5*((\x-8.25)/0.23)^2)});
\end{scope}
\draw[->] (-0.4,0) -- (10.4,0) node[right] {$x_1$};
\draw[dashed] (2.3,0) -- (2.3,2.35);
\draw[dashed] (8.25,0) -- (8.25,2.0);
\draw[<->] (2.3,-0.45) -- node[fill=white,inner sep=1pt] {$L$} (8.25,-0.45);
\draw[<->] (7.95,0.4) -- node[fill=white,inner sep=1pt] {$r$} (8.55,0.4);
\node[blue] at (2.3,2.6) {$h$};
\node[red] at (8.25,2.25) {$\eta\,g_{r,L}$};
\node at (5.3,3.0) {$F_{r,\eta,L}=(1-\eta)h+\eta g_{r,L}$};
\end{tikzpicture}
\caption{Schematic of the separated-mixture construction in Theorem \ref{prop:theta-dichotomy}(b).  The broad background controls $\Ical$, while a tiny narrow bump with mass $\eta$, scale $r$, and separation $L$ can dominate $\Qcal$ and $\Dcal$ after rescaling.  Both component profiles have rapidly decaying, noncompact tails.}
\label{fig:dichotomy-mixture}
\end{figure}
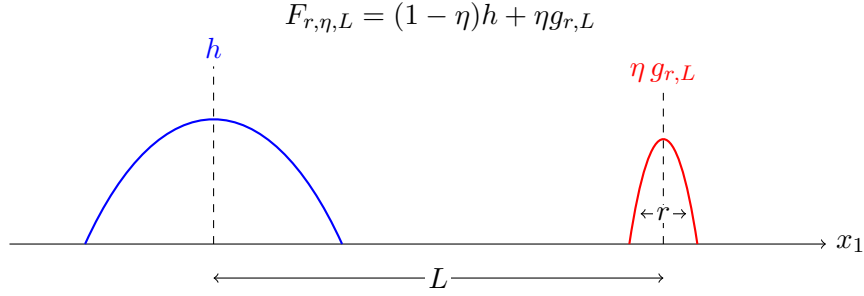

\begin{proof}
Part (a) is immediate: $\Ical[f]\ge0$ for every admissible density $f$, and $\Qcal[f]^2>0$ whenever $\Qcal[f]>0$, so universal nonnegativity of $\Dcal[f]$ implies
\[
        \frac{\Ical[f]\Dcal[f]}{\Qcal[f]^2}\ge0
\]
for every admissible $f$ in every dimension $d\ge2$.  Hence $\theta_d^*\ge0$ for all $d\ge2$, and therefore $0\le\theta_\infty^*$.  The strict upper bound $\theta_\infty^*<1$ still follows from \eqref{eq:theta-less-one}.

For part (b), fix such a density $g$ on $\R^{d_0}$ and let $h(x)=(2\pi)^{-d_0/2}e^{-|x|^2/2}$ be the standard Gaussian.  We use three parameters, as in Figure \ref{fig:dichotomy-mixture}: $r>0$ is the spatial scale of the bad bump, $\eta\in(0,1)$ is its total mass, and $L>0$ is the separation distance from the background.  For $r>0$, define
\[
        g_r(x)=r^{-d_0}g(x/r).
\]
Since
\[
        \log g_r(x)=\log g(x/r)-d_0\log r,
\]
we have
\[
        \Hmat_{g_r}(x)=r^{-2}\Hmat_g(x/r),
        \qquad
        \nabla\Hmat_{g_r}(x)=r^{-3}(\nabla\Hmat_g)(x/r).
\]
Changing variables therefore gives the scaling laws
\begin{equation}\label{eq:negative-D-scaling}
        \Ical[g_r]=r^{-2}\Ical[g],
        \qquad
        \Qcal[g_r]=r^{-4}\Qcal[g],
        \qquad
        \Dcal[g_r]=r^{-6}\Dcal[g].
\end{equation}

Fix $\eta\in(0,1)$ and $L>0$, set
\[
        g_{r,L}(x)=g_r(x-Le_1),
        \qquad
        F_{r,\eta,L}(x)=(1-\eta)h(x)+\eta g_{r,L}(x),
\]
and write
\[
        J_1[f]=\tr(\Hmat_f)f,
        \qquad
        J_2[f]=\tr(\Hmat_f^2)f,
        \qquad
        J_3[f]=\Big(|\nabla\Hmat_f|^2+2\tr(\Hmat_f^3)\Big)f
\]
for the pointwise integrands of $\Ical$, $\Qcal$, and $\Dcal$.  Each $J_\nu[f](x)$ depends smoothly on the value of $f$ and its partial derivatives up to order three at $x$ whenever $f(x)>0$, and satisfies the homogeneity relation
\[
        J_\nu[cf]=c\,J_\nu[f]
        \qquad (c>0),
\]
because $\Hmat_{cf}=\Hmat_f$ and $\nabla\Hmat_{cf}=\nabla\Hmat_f$.

We first establish the asymptotic additivity formulas.  For each fixed $r>0$ and $\eta\in(0,1)$,
\begin{equation}\label{eq:separated-mixture-I}
        \Ical[F_{r,\eta,L}]
        =
        (1-\eta)\Ical[h]+\eta\Ical[g_r]+o_L(1),
\end{equation}
\begin{equation}\label{eq:separated-mixture-Q}
        \Qcal[F_{r,\eta,L}]
        =
        (1-\eta)\Qcal[h]+\eta\Qcal[g_r]+o_L(1),
\end{equation}
and
\begin{equation}\label{eq:separated-mixture-D}
        \Dcal[F_{r,\eta,L}]
        =
        (1-\eta)\Dcal[h]+\eta\Dcal[g_r]+o_L(1)
\end{equation}
as $L\to\infty$.  This is the only point at which the separation parameter $L$ is used.  To prove these relations, fix $\nu\in\{1,2,3\}$ and $\varepsilon>0$.  Since $J_\nu[h]$ and $J_\nu[g_r]$ are integrable, we may choose $R$ so large that
\[
        \int_{|x|>R}|J_\nu[h](x)|\dd x
        +
        \int_{|x|>R}|J_\nu[g_r](x)|\dd x
        <\varepsilon.
\]
For $L>4R$, the balls $B(0,R)$ and $B(Le_1,R)$ are disjoint.  On $B(0,R)$, rapid decay of $g_r$ implies that $g_{r,L}$ and its derivatives up to order three are $O_N(L^{-N})$ for every $N$, uniformly in $x\in B(0,R)$.  Since $h$ is strictly positive on the compact ball $B(0,R)$, smooth dependence on the jet of the density gives
\[
        J_\nu[F_{r,\eta,L}](x)
        =
        (1-\eta)J_\nu[h](x)+O_N(L^{-N})
        \qquad \text{uniformly on }B(0,R).
\]
After translating by $Le_1$, the same argument on $B(Le_1,R)$ gives
\[
        J_\nu[F_{r,\eta,L}](x)
        =
        \eta J_\nu[g_r(x-Le_1)]+O_N(L^{-N})
        \qquad \text{uniformly on }B(Le_1,R).
\]
On the complement of these two balls, the rapid decay of $h$ and $g_{r,L}$, together with the smooth rational dependence of $J_\nu$ on the jet of the density, provides an integrable majorant independent of $L$.  Its integral can be made $O(\varepsilon)$ by the choice of $R$.  Integrating over the three regions, then sending $L\to\infty$ and finally $\varepsilon\downarrow0$, proves \eqref{eq:separated-mixture-I}--\eqref{eq:separated-mixture-D}.

With the additivity formulas in hand, choose $\eta=r^3$.  This choice keeps $\Ical$ at order one while forcing $\Qcal$ and $\Dcal$ to scale like $r^{-1}$ and $r^{-3}$, respectively.  For each sufficiently small $r>0$, choose $L=L(r)$ large enough that the three $o_L(1)$ terms above are $o_r(1)$ as $r\downarrow0$.  Then \eqref{eq:negative-D-scaling} yields
\[
        \Ical[F_{r,r^3,L(r)}]
        =
        \Ical[h]+r\,\Ical[g]+o_r(1)
        =
        \Ical[h]+o_r(1),
\]
\[
        \Qcal[F_{r,r^3,L(r)}]
        =
        \Qcal[h]+r^{-1}\Qcal[g]+o_r(1)
        =
        r^{-1}\Qcal[g]+O(1),
\]
and
\[
        \Dcal[F_{r,r^3,L(r)}]
        =
        \Dcal[h]+r^{-3}\Dcal[g]+o_r(1)
        =
        r^{-3}\Dcal[g]+O(1).
\]
Therefore
\[
        \frac{\Ical[F_{r,r^3,L(r)}]\Dcal[F_{r,r^3,L(r)}]}
             {\Qcal[F_{r,r^3,L(r)}]^2}
        =
        \frac{\bigl(\Ical[h]+o_r(1)\bigr)\bigl(r^{-3}\Dcal[g]+O(1)\bigr)}
             {\bigl(r^{-1}\Qcal[g]+O(1)\bigr)^2}
        =
        \frac{\Ical[h]\Dcal[g]}{\Qcal[g]^2}\,r^{-1}+o(r^{-1}).
\]
Since $\Dcal[g]<0$, the right-hand side tends to $-\infty$ as $r\downarrow0$.  Hence $\theta_{d_0}^*=-\infty$.  Finally, Proposition \ref{prop:theta-monotone} implies $\theta_d^*=-\infty$ for every $d\ge d_0$.
\end{proof}

By Lemma \ref{lem:identities}, the sign condition $\Dcal[f]\ge0$ is equivalent to convexity of Fisher information along heat flow.  In other words, the first barrier to any nontrivial lower bound on $\theta_d^*$ is the ordinary convexity problem for Fisher information itself.  If one could prove $\Dcal[f]\ge0$ for all admissible densities in all dimensions, then Theorem \ref{prop:theta-dichotomy}(a) would give $\theta_\infty^*\in[0,1)$.  If one could instead find a single admissible density with $\Dcal[f]<0$, then Theorem \ref{prop:theta-dichotomy}(b) would force the entire high-dimensional problem to collapse to $\theta_\infty^*=-\infty$.

At present, the known status is dimension-dependent.  In dimension one, $\Dcal[f]\ge0$ follows from the log-convexity theorem of \cite{LedouxNairWang}.  In dimension two, it follows already from the square-root convexity theorem of Liu and Gao \cite{LiuGao2023}.  More generally, Guo, Yuan, and Gao proved the completely monotone conjecture for order $\ell=2$ in dimensions $d=2,3,4$ \cite{GuoYuanGao2022}; see also the review of Ledoux \cite{LedouxReview}, where this result is recorded together with the identity
\[
        I''(t)=\int_{\R^n} f_t\Bigl(|\nabla^3 v_t|^2-2T_3(v_t)\Bigr)\dd x,
        \qquad v_t=\log f_t,
\]
where
\[
        T_3(v):=\sum_{i,j,k=1}^n
        (\partial_{ij}v)(\partial_{ik}v)(\partial_{jk}v)
        =\tr\bigl((\nabla^2v)^3\bigr).
\]
Thus the right-hand side is exactly our functional $\Dcal[f_t]$.  Consequently, $\Dcal[f]\ge0$ is currently known for all admissible densities in dimensions $d=1,2,3,4$.  For $d\ge5$, we are not aware of either a proof or a counterexample, and the sign problem should therefore be regarded as open.

\subsection{Open problems}

Theorem \ref{prop:theta-dichotomy} shows that the large-dimensional problem has a basic two-way structure.  Either one finds an admissible density with $\Dcal[f]<0$, in which case $\theta_\infty^*=-\infty$, or else the asymptotic problem remains on the nonnegative branch.  Within that nonnegative branch, the present results do not determine whether the optimal dimension-free lower bound degenerates to
\[
        \theta_\infty^*=0,
\]
or instead remains strictly positive and converges to a value in the interval $(0,1)$.  The generalized hexagonal frequency-lattice calculation above suggests that structured higher-dimensional families may be relevant, but it does not determine the limiting value of $\theta_\infty^*$.

At present, the basic unresolved questions are the following.
\begin{enumerate}[label=(\roman*)]
\item Determine $\theta_2^*$ more precisely inside the rigorous interval \eqref{eq:theta-two-range}.
\item Decide whether the monotonicity in \eqref{eq:theta-monotone} is strict for some, or all, dimensions.
\item Determine whether $\Dcal[f]$ can ever be negative.  Equivalently, by Lemma \ref{lem:identities}, determine whether Fisher information is always convex along heat flow.  By Theorem \ref{prop:theta-dichotomy}, a single admissible example with $\Dcal[f]<0$ would force $\theta_\infty^*=-\infty$.
\item Assuming the negative-\(\Dcal\) alternative in Theorem \ref{prop:theta-dichotomy} does not occur, determine the optimal dimension-free constant $\theta_{\mathrm{univ}}\in[0,1)$ such that
\[
        \Ical[f]\Dcal[f]\ge \theta_{\mathrm{univ}}\Qcal[f]^2
\]
for every $d\ge2$ and every admissible density $f$ on $\R^d$.  Equivalently, determine whether the asymptotic sharp constant satisfies $\theta_\infty^*=0$ or $\theta_\infty^*\in(0,1)$.
\end{enumerate}

Collecting the rigorous results currently available, we have
\[
        \theta_1^*=1,
        \qquad
        \frac12\le\theta_2^*<1,
        \qquad
        \theta_{d+1}^*\le\theta_d^*<1\quad(d\ge2),
\]
and hence
\[
        \theta_\infty^*=\inf_{d\ge2}\theta_d^*\in[-\infty,1).
\]
In view of Theorem \ref{prop:theta-dichotomy}, this leaves only two global possibilities:
\[
        \theta_\infty^*\in[0,1)
        \qquad\text{or}\qquad
        \theta_\infty^*=-\infty.
\]
At present, the nonnegative alternative is rigorously verified only through dimension four.  The exact values of $\theta_d^*$, the sign of $\Dcal$ in dimensions $d\ge5$, the existence of a dimension-free lower bound, and the asymptotic behavior of $\theta_\infty^*$ remain open.

\section{Conclusion}

The present paper settles the fate of Fisher-information log-convexity along heat flow: it holds in dimension one by \cite{LedouxNairWang} and fails in every dimension $d\ge2$ by Theorem \ref{thm:main} and Corollary \ref{cor:high-d}.  Through the implication chain \eqref{eq:hierarchy}, this also closes the multidimensional versions of GCMC, McKean's conjecture, and the entropy power conjecture.

What remains is the quantitative problem of understanding the sharp constants $\theta_d^*$.  Section \ref{sec:high-d} shows that the two-dimensional mechanism extends to higher dimensions both by tensorization and by a frequency-lattice generalization of the hexagonal zero-mode mechanism in $\R^d$, while Section \ref{sec:open-problems} proves that $\theta_1^*=1$, that $\theta_{d+1}^*\le\theta_d^*$, and that the first barrier to any global lower bound is the sign of $\Dcal$.  The main open directions are therefore to determine whether $\Dcal$ can become negative in high dimension, to sharpen the bounds on $\theta_d^*$, and to understand whether structured families of this kind capture the true asymptotic behavior of $\theta_\infty^*$.


\end{document}